\documentclass[11pt,preprint]{article}

\usepackage{geometry}
\usepackage[comma]{natbib}
\usepackage{url}
\usepackage[utf8]{inputenc}
\usepackage{csquotes}
\usepackage{graphicx}
\usepackage{amssymb,amsthm,amsmath}
\usepackage{enumitem}
\usepackage{hyperref}
\usepackage{breakurl}
\usepackage{psfrag}
\usepackage{color}
\usepackage{todonotes}
\usepackage[capitalise,noabbrev]{cleveref} 

\usepackage{algorithmicx}
\usepackage[ruled]{algorithm}
\usepackage{algpseudocode}

\usepackage{authblk}

\theoremstyle{theorem}
\newtheorem{theorem}{Theorem}
\newtheorem{definition}[theorem]{Definition}

\newtheorem{example}[theorem]{Example}
\newtheorem{corollary}[theorem]{Corollary}

\newtheorem{lemma}[theorem]{Lemma}

\newcommand{\x}{x}
\newcommand{\coloneqq}{:=}
\newcommand{\kl}[1]{\left(#1\right)}
\newcommand\skl[1]{(#1)}
\newcommand\mkl[1]{\big(#1\bigr)}

\newcommand\set[1]{\left\{#1\right\}}
\newcommand\sset[1]{\{#1\}}

\newcommand\norm[1]{\left\Vert#1\right\Vert}

\newcommand\snorm[1]{\lVert#1\rVert}
\newcommand\mnorm[1]{\bigl\lVert#1\bigr\rVert}

\newcommand\abs[1]{\left\lvert#1\right\rvert}
\newcommand\sabs[1]{\lvert#1\rvert}
\newcommand\inner[2]{\left\langle  #1,#2 \right\rangle}

\newcommand\sinner[2]{\langle  #1,#2 \rangle}

\newcommand\R{\mathbb{R}}

\newcommand\N{\mathbb{N}}

\newcommand\F{\mathcal{F}}

\newcommand\rf{R}

\newcommand{\la}{\lambda}
\newcommand{\La}{\Lambda}
\newcommand{\Om}{\Omega}
\newcommand{\eps}{\varepsilon}

\DeclareMathOperator{\wk}{\mathbf P}

\DeclareMathOperator*{\argmin}{argmin}

\DeclareMathOperator{\samp}{\boldsymbol S}
\DeclareMathOperator{\iop}{\boldsymbol E}

\DeclareMathOperator*{\minimize}{min}
\DeclareMathOperator*{\st}{s.t.}

\newcommand{\Id}{\boldsymbol{I}}

\newcommand{\U}{\mathcal U}
\newcommand{\Fs}{\mathcal F}
\newcommand{\X}{\boldsymbol X}
\newcommand{\rE}{\boldsymbol E}
\newcommand{\Po}{\mathbf P}

\DeclareMathOperator{\A}{\boldsymbol K}

\DeclareMathOperator{\KK}{\mathcal A}
\newcommand{\kk}{a}
\DeclareMathOperator{\T}{\boldsymbol \Pi}
\DeclareMathOperator{\base}{\boldsymbol \Phi}

\newcommand{\figref}[1]{\textbf{Figure~\ref{#1}}}

\begin{document}

\title{A Variational View on Statistical Multiscale Estimation}

\author{Markus Haltmeier} 
\affil{Department of Mathematics, University of Innsbruck\\
Technikestra{\ss}e 13, A-6020 Innsbruck, Austria\\
Email: {\tt markus.haltmeier@uibk.ac.at}}

\author{Housen Li, and Axel Munk}
\affil{Institute of Mathematical Stochastics, University of G\"ottingen\\
Goldschmidtstr.\ 7,  37077 G\"ottingen, Germany\\
Email: {\tt housen.li@mathematik.uni-goettingen.de}, {\tt munk@math.uni-goettingen.de}} 

\maketitle

\begin{abstract}
We present a unifying view on various statistical estimation techniques including {penalization}, variational and thresholding methods. These estimators will be analyzed in the context of statistical linear inverse problems including nonparametric and change point regression, and high dimensional linear models as examples.  
 Our approach reveals many seemingly unrelated estimation schemes as special instances of a general class of variational multiscale estimators, named MIND (MultIscale Nemirovskii--Dantzig). These estimators {result from minimizing certain regularization} functionals under convex constraints that can be seen as multiple statistical tests for local hypotheses. 
 
For computational purposes, we recast MIND in terms of simpler unconstraint optimization problems via  Lagrangian penalization as well as Fenchel duality. Performance of several MINDs is demonstrated on numerical examples. 


\end{abstract}

\tableofcontents

\section{Introduction}
Recovering a (typically high dimensional) parameter vector $\beta \in \R^p$ from (possibly indirect) noisy observations  is a fundamental task in {modern data analysis and has been a long standing topic of intense investigation in} statistical, mathematical, and related sciences.  Applications are vast. What we have primarily in mind are imaging and signal recovery problems, such as tomography or spectroscopy. Efficient recovery of these signals (we use the term \emph{signal} in the following for a general $\beta$, including images as well) always relies on structural prior information, which is in many cases given by neighborhood information on the structure of the signal. In the context of imaging a neighborhood often is with respect to spatial distance, e.g.\ expressed in certain smoothness assumptions. In other scenarios (e.g.\ graphs) this should be understood as \enquote{structurally similar regions}. Other assumptions may concern certain features such as peaks of a signal or texture of an image.  Although a serial signal is of physical space dimension one and an image of dimension two or three,  the mathematical  \enquote{effective dimension} of the problem is given through the complexity of the modeling function system of the signal. For example, if a one-dimensional signal is sampled at $n$ points and is assumed to be potentially different at each sampling point, the effective dimension $p$ may be thought of $p=n$. Related to this are problems that are \enquote{truly} high dimensional in the sense that potentially many more coefficients (or parameters) have to be estimated than observations available, {viz.\ $p\gg n$.} Applications include high throughput data in various \enquote{omics} technologies in genetics or large scale networks, to mention a few. 

Probably the most prominent and simplest unifying model in this context assumes a linear relationship $\X \in \R^{n \times p}$ between the unknown  coefficient vector $\beta \in \R^p$ and the observations $Y = \kl{Y_1, \dots, Y_n}$,
\begin{equation}  \label{eq:ip}
	 Y_i
	 =
	 \kl{\X  \beta}_{i}
	 +
	 \eps_i
	 \qquad \text{ for }
	 \;
	  i  = 1,  \dots,  n \,.
\end{equation}
Here  $n$ is  the number of observations, $p$ is the number of unknown parameters, $ Y_i$ are the observed data, and $\eps_i$'s model the error (noise) in the observations. Throughout this paper we assume for the sake of simplicity a white noise error, i.e.\ the $\eps_i$'s  are independent and  normally distributed random variables $\mathcal{N} \skl{0, \sigma^2}$ with mean zero and variance $\sigma^2$. We stress, however, that much of the following can be extended  to other models, e.g.\ to generalized linear exponential family regression, and to data with certain dependencies {or heterogeneous errors.} 

{
Some convention of notation is as follows: Matrices and linear operators are in boldface. Vectors are in the column form and are written in normal font. For a vector $v = (v_i)_{i = 1}^n$ in $\R^n$, we denote the $\ell^0$-quasi-norm by $\norm{v}_0 = \# \set{i : v_i \neq 0, i = 1, \ldots, n}$, the $\ell^2$-norm by $\norm{v}_2 = (\sum_{i=1}^n v_i^2)^{1/2}$, and the $\ell^{\infty}$-norm by $\norm{v}_\infty = \max_{1 \le i \le n} \abs{v_i}$. In case of $n = p$, i.e.\ the number of observation equals the number of parameters, we always use $n$ in place of $p$. Let  $x_+= \max \set{x,0}$ denote the positive part of a real variable $x$. 
}

\subsection{The linear model: Examples}

The linear model in~\eqref{eq:ip}  contains many important instances of statistical models. Concrete examples we shall consider are as follows.

\begin{enumerate}[label=(\Roman*)]
\item
\textsc{Non-parametric regression.\label{M1}}
Here we take {$p = n$, $\X = \Id$ as the $n$-dimensional identity matrix},
and the unknown parameter
\begin{equation*}
\beta = \samp_n \kl{f}  \coloneqq  \kl{f \skl{x_{i}}}_{i=1}^n
\end{equation*}
as values of a regression function $f \colon [0,1]^d \to \R$ at sampling points $x_{i}\in [0,1]^d$ for some dimension $d \ge 1$ (here the $d$-dimensional unit cube is chosen for simplicity). Given {the vector} $\beta$, the full regression function $f$ can be recovered by some interpolation or approximation procedure that relies on additional structural assumptions on $f$, encoded in a function space $\mathcal{F}$,  such as certain smoothness properties. These can be often expressed conveniently in terms of approximation properties with respect to a certain function system, such as radial basis functions \citep{Wen05}, splines \citep{Wah90}, polynomials, trigonometric series \citep{WaSh01}, or wavelets \citep{Mal09}. Many estimators do not explicitly reflect this decomposition and estimation of $\beta$ and interpolation
are performed simultaneously.

\item \textsc{Linear  inverse problems.\label{M2}}
In this example, the unknown $\beta = \samp_p \kl{f} = \kl{f \skl{x_j}}_{j=1}^p$
consists of  values of an element in some function space $\F$ as in model \ref{M1}. It is linked to the observational vector by a \emph{system matrix} 
\begin{equation*}
	\X = \samp_n \circ \A \circ \iop_p \colon \R^p \to \R^n,
\end{equation*}
which is the composition of  an abstract interpolation operator
$\iop_p \colon \R^p \to \F$, a  linear operator
$\A  \colon \F \to \U$ modeling the particular inverse problem,
and  another  sampling operator $\samp_n \colon \U \to \R^n$  (with a slight abuse of notation)
on the model space $\U \supset \A (\F)$. Prominent examples include
the Radon transform in computed tomography \citep{Don95b,NaWu01}, the Fourier transform in magnetic resonance tomography \citep{Eps08} or convolution with a point spread function in optics \citep{BeBDV09,AsEM15}, ranging from astronomical imaging to high resolution microscopy. For $\A$ the identity, we obtain the regression model \ref{M1}.

\item \textsc{Change point detection.\label{M3}}
We consider $f\colon [0,1] \to \R$ as in model \ref{M1} with $d = 1$, and  take $\beta_i  =  f \kl{x_i} - f \kl{x_{i-1}}$,
for $i  \in \set{2, \dots , n}$, as the
jump sizes of the piecewise constant regression function
$f $ sampled at the locations $x_i$. Further we denote by
$ \beta_1 = f \kl{x_1}$ its offset.
The function $f$ then can be recovered from
$\beta = \kl{\beta_i}_{i=1}^n$  by $f \kl{x_i} = \sum_{k=1}^i \beta_k$  for $i = 1, \dots, n  \,.$
The relation between the jump sizes $\beta$ and data
$Y_i$ can be written in the form of \eqref{eq:ip}, where
the system matrix $\X$ takes the form  \citep{BoyKemLieMunWit09}
\begin{equation*}
\X
=
\begin{pmatrix}
1  & 0 & \cdots & 0 \\
\vdots  &  \ddots   &   & \vdots \\
\vdots  &  &   \ddots&   0\\
1 & \cdots &   \cdots  & 1
\end{pmatrix} \in \R^{n \times n}\,.
\end{equation*}
Taking $p=n$ reflects the fact that no a-priori
assumptions on the number and locations of  the jumps in
this model have been posed. Neighboring information amounts to information on length of connected segments, i.e.\ how non-zero $\beta_i$'s are located.

\item \textsc{High dimensional regression.\label{M4}}
Here the  unknown $\beta \in \R^p$ is a high dimensional
parameter vector, typically  $p \gg n$, which means that the number of unknown parameters is  much larger than the number of observations. In contrast to the first two examples, {in general,}
no neighboring structure on $\beta$  is posed, rather  a sparsity assumption {(only a few of the $p$ coefficients $\beta_i$ are nonzero)} that constrains  the set of possible  solutions (see e.g.\ \citealp{BueGee11,Wai19} and the references therein).
\end{enumerate}

{Lastly, we stress that many models can be cast in the form of \eqref{eq:ip}, but the error may not be an independent Gaussian white noise (thus beyond the scope of this paper). For example, consider the errors-in-variable (or measurement error) model given by
\begin{equation}\label{eq:eivm}
Y_i = (\X \beta)_i + \eps_i,\quad i = 1, \ldots, n, \quad \text{ and } \quad \tilde \X = \X + \rE\; \in\; \R^{n\times p},
\end{equation}
where we assume for simpilicity that $\eps_i$'s are independent $\mathcal{N}(0, \sigma^2)$ with variance $\sigma^2$, entries of $\rE$ are independent $\mathcal{N}(0, \tau^2)$ with variance $\tau^2$, and $\eps_i$'s and  entries of $\rE$ are independent. 
The aim is to recover the unknown parameter $\beta$ given observations $Y = (Y_i)_{i=1}^n$ and $\tilde \X$ (which is a perturbed version of~$\X$). \eqref{eq:eivm} can be rewritten as 
$$
Y_i = (\tilde\X \beta)_i + \tilde\eps_i \qquad \text{with}\quad \tilde\eps_i = \eps_i - (\rE\beta)_i.
$$
Note that $\tilde\eps_i$'s are still independent and Gaussian, more precisely, $\mathcal{N}(0,\, \sigma^2 + \tau^2\norm{\beta}_2^2)$, but their variance now depend on the unknown parameter $\beta$ (see e.g.\ \citealp{CRSC06}).}

\subsection{Types of variational estimation}

We introduce here two prominent types of variational methods for estimating the parameter
$\beta$ in~\eqref{eq:ip}, and examples will be given in subsequent sections.
 
Among the best known approaches is the method of maximum likelihood, which boils down to least squares estimation in our setup. It minimizes the residual sum of squares
\begin{equation}\label{eq:rss}
	C_Y(\beta) := \norm{\X \beta - Y}^2_2	
	=
	\sum_{i=1}^n \abs{ \kl{\X  \beta}_i -  Y_i }^2 \,.
\end{equation}
If $\beta$ is a low dimensional vector, the maximum likelihood estimator
is well known  to be asymptotically normal and efficient under proper regularity conditions on $\X$ (e.g.\ \citealp{vdVa98}).
Note, however, that model \ref{M3} is a notable exception {(see \citealp{KoKo11}).} 

In case that $\beta$ represents many degrees of freedom, minimizing the residual sum of squares leads to over-fitting of the data. See e.g.\ \citet{Port88} for conditions on $p$ when the maximum likelihood estimator fails to be asymptotically normal. In a certain sense, this is even true, when $\beta$ is assumed to have only a few nonzero coefficients (sparseness) because the model selection error is not uniformly controllable \citep{LePo06}, {without further assumptions}. As the system matrix $\X$ becomes more ill-conditioned (i.e.\ one or more eigenvalues of $\X^\mathsf{T} \X$ are close to zero), all this becomes even more critical because the fluctuations of $\beta$ are heavily damped through $\X$ and the reconstruction process becomes more unstable (see \citealp{Osu86} for some examples).

Hence, any reasonable  estimation procedure for a high dimensional
vector $\beta$  has, either implicitly or explicitly, to account
for additional properties of the unknown parameter,
such as smoothness, sparsity, or other structural information (recall models \ref{M1}--\ref{M4}), i.e.\ to \emph{regularize} the solution.
Such a-priori information can be incorporated by requiring the
value of a \emph{regularization functional} $\rf \colon \R^p \to \R \cup \set{\infty}$ to be small
at the particular estimate. In the following, we discuss two types of approaches, at a first glance, seemingly unrelated. We are aware that our treatment is {incomplete}. For example, there is a rapidly growing literature on (nonparametric) Bayesian techniques to advise such regularization as well (see e.g.\ \citealp{GhovdV17}), which is, however, beyond the scope of this survey. 

\subsubsection{Penalized estimation}\label{sss:pe}
Probably the most prominent approach to include such structural information on
$\beta$ into the estimation process is to incorporate a \emph{regularization functional} $R$  into a \emph{data fidelity
term} $G(\X \beta; Y )$ that measures the discrepancy from the data, typically as an additive penalty. 
The penalized estimator results from a solution of the Lagrangian variational problem
\begin{equation} \label{eq:rss-pen}
\minimize_{\beta \in \R^p}G(\X \beta; Y ) +  \gamma  \rf\skl{\beta},
\end{equation}
for some $\gamma > 0$.  In case of $G(\X \beta; Y) = C_Y(\beta)$ in \eqref{eq:rss}, this is known as the \emph{penalized least squares} estimator, {which is a particular case of the log-likelihood function of the model. General $G$ results in e.g.\  \emph{penalized maximum-likelihood} estimation (see e.g.\ \citealp{EggLar09}). Solution of \eqref{eq:rss-pen}} yields regularized estimates with regularity measured in terms of the regularization functional $R$. A plethora of estimators follow this paradigm, some instances of $R$ will be discussed later on. 

{The proper choice of penalty parameter $\gamma$ in \eqref{eq:rss-pen} relies on the precise structure of the signal, which is not available in practice. Thus, a data driven strategy for the selection of $\gamma$ is needed, and turns out to be a challenging and delicate issue. A wide range of methods have been suggested in the literature. Details are far beyond the scope of this paper, so we just mention Mallows' $C_p$ \citep{Mal00,LiWe20},  cross-validation \citep{All74,Sto74}, generalized cross-validation \citep{Wah77}, plug-in techniques \citep{Loa99}, bootstrap based methods \citep{Bre92,Sha96}, and techniques that are built on the Lepski\u{\i} balancing principle \citep{Lep90,LeMS97}, to name only a few. However, in the following, we present a constrained formulation of \eqref{eq:rss-pen}, which offers a statistically simpler strategy of the corresponding threshold parameter selection and thereby circumvents the difficulties encountered with the choice of $\gamma$ to some extent. }

\subsubsection{Constrained estimation} 
A seemingly different approach to incorporate $R$ into the estimation process
is based on the idea to use the data fidelity term as a constraint, resulting in the \emph{constrained} estimator
\begin{equation} \label{eq:rss-const}
	\left\{
	\begin{aligned}
	&\minimize_{\beta \in \R^p}  &&  \rf\skl{\beta}   \\
	&\st &&
	G(\X \beta; Y)
	\leq   q
	\;,
	\end{aligned}
	\right.
\end{equation}
for some threshold $q > 0$. It reduces to the \emph{constrained least squares} estimator when $G(\X \beta; Y) = C_Y(\beta)$ in \eqref{eq:rss}.  In fact, it is easily seen and well known that constrained estimation is closely related to penalized estimation in \eqref{eq:rss-pen}, see \cref{sec:variational}.  However, even in the simple case of $G(\X \beta; Y) = C_Y(\beta)$, the correspondence between the two parameters $\gamma$ in \eqref{eq:rss-pen} and $q$ in \eqref{eq:rss-const} is not given explicitly  and depends on the data $Y$.  It is exactly the lack of this explicit correspondence that makes the different nature of these estimators, as the (data driven) choice of $\gamma$ in \eqref{eq:rss-pen} and $q$ in \eqref{eq:rss-const} is difficult to \enquote{translate} from the constrained formulation into the penalized one and vice versa. {As mentioned, the choice of $q$ in \eqref{eq:rss-const} is comparably easier than that of $\gamma$ in \eqref{eq:rss-pen}, because a proper choice of $q$ depends essentially on the distribution property of the error $\eps_i$ in \eqref{eq:ip}, which is usually (approximately) available. One exemplary method is Morozov's discrepancy principle \citep{Mor66}, see also \cref{sec:MIND}.}

In the following, some prototypical estimators will be introduced.

\subsection{Examples and spatial adaptation}

{In this section, we work with the general model in \eqref{eq:ip}, unless explicitly specified.}

\subsubsection{Smoothing splines}\label{sss:ss}

We start with a prominent instance of the penalized least squares estimator, a \emph{smoothing spline}  \citep{Wah90}. The regularization functional $R$ is the squared $\ell^2$-norm of the (discrete) derivative,
\begin{equation*}
	\rf\skl{\beta}
	=
	\frac{1}{2}
	\sum_{i=1}^{p-1} \abs{ \beta_{i+1}-\beta_i}^2 \,,
\end{equation*}
or some higher order analog. Increasing the threshold $q$ in \eqref{eq:rss-const} (or equivalently the penalty $\gamma$ in \eqref{eq:rss-pen}) yields smoother estimates (i.e.\ solutions of \eqref{eq:rss-pen} and \eqref{eq:rss-const}) and therefore its particular choice sensitively affects the regularity of the resulting estimator. 

\subsubsection{Spatial adaptation}
{Recall that $\beta = \samp_p \kl{f} =  \kl{f \skl{x_{i}}}_{i=1}^p$.}  The parameters $\gamma$ and $q$ in the spline approach in \cref{sss:ss} act globally over the domain of $f$ but the regularity of a regression function $f$ is often \enquote{spatially} varying, {i.e.\ it will depend on $ x\in [0,1]^d$.} This suggests that better results can be obtained by introducing \emph{local weights} $w_i>0$  depending on the spatially varying smoothness, i.e.
\begin{equation*}
	\rf_w \kl{\beta}
	=
	\frac{1}{2}
	\sum_{i=1}^{p-1} w_i \abs{ \beta_{i+1}-\beta_i}^2 \,.
\end{equation*}
Because the local smoothness of the underlying function is unknown one faces the problem of how actually choosing the weights $w_i$ in an adaptive (data driven) manner. This  turns out to be a difficult task as the minimization problems in \eqref{eq:rss-pen} and \eqref{eq:rss-const} when additionally optimizing over $w$ are  not convex anymore, {in general.} Further, identification of the actual estimator for $\beta$ and the weights $w_i$ is in general not possible. A similar comment applies to the attempt to \enquote{localize} other global regularization functionals such as the total variation, although some methods have been suggested (e.g.\ \citealp{DavKov01,DoHR11,LeBe15}). In the following we will therefore give a reformulation of this attempt in a more general context that avoids these difficulties and provides feasible estimators.

\subsubsection{Wavelet soft-thresholding}
Wavelet (and more generally dictionary) based thresholding methods have been proven to provide certain spatial adaptivity without explicitly including spatially varying weights (e.g.\ \citealp{DonJoh94,DonJoh95}). Heuristically speaking, the reason is that the spatial variability has already been incorporated in the basis (or in general dictionary) functions and is respected by the thresholding procedure. This is an important feature, not shared by other series estimators, such as Fourier estimators, which are not localized in time domain (see e.g.\  \citealp{Har97,WaSh01,Tsy09}).  As an introductory example it is illustrative to represent the soft-thresholding wavelet estimator as a constraint estimator as in \cite{Don95}. Assume the nonparametric regression model~\ref{M1}, i.e.\ $\X = \Id$. In its constraint formulation
 the coefficient vector of the soft-thresholded wavelet estimator is obtained as the (unique) minimizer of
$\norm{\beta}_2$ among all  $\beta \in \R^n$
satisfying the constraint
\begin{equation}\label{eq:maxo}
\max_{\la \in \La} \abs{ \inner{\phi_\la}{Y -  \beta} } \leq   q
\end{equation}
(see \cref{thm:soft} in \cref{sec:threshold}).
Here $(\phi_{\lambda})_{\lambda \in \Lambda}$ denotes a system of wavelets or some other system spanning  $\R^n$ (see \cref{sec:threshold} {for examples}). This optimization problem obviously falls in the framework of \eqref{eq:rss-const}. The
residuals $Y - \beta$ are analyzed by the $\ell^\infty$-norm {(i.e.\ the maximum of absolute values in \eqref{eq:maxo})} in the wavelet domain. 

\subsubsection{The Dantzig selector}
The \emph{Dantzig selector} \citep{CanTao07} for \eqref{eq:ip} is
defined as a solution of the optimization problem 
\begin{equation} \label{eq:Dantzig}
	\left\{
	\begin{aligned}
	&\minimize_{\beta \in \R^p}  &&   \sum_{j=1}^p  \abs{ \beta_j }   \\
	&\st &&
	\max_{1 \leq i \leq {p}}
	\abs{ \inner{\X_{i}}{ Y - \X \beta } }
	\leq   q
	\;,
	\end{aligned}
	\right.
\end{equation}
with  $\X_{i}$ denoting  the $i$-th column of  $\X$. Under the assumption that $\X$ satisfies the \emph{restricted isometry property} and $\beta$ is \emph{sparse}, \cite{CanTao07} showed that  $\snorm{ \hat \beta  - \beta}_2^2$ is, with high probability, bounded by a logarithmic quantity times the oracle risk $\sum_{i=1}^p \min \sset{\beta_i^2,  \sigma^2}$.

The Dantzig selector acts on scales which depend on the system matrix $\X$ itself.  It is therefore not necessarily  multiscale in nature: For example in case that $\X$ is the identity matrix {(model \ref{M1})} only the smallest scales are taken into account, i.e.\ {the constraint acts only} on each single observation. Note the difference to the constrained least squares estimators where the side constraint acts only on the largest scale. {Hence, these estimators measure the data fidelity in a complementary way} from a statistical point of view. {Both estimators} can be  extended to a truly multiscale estimator, {which additionally takes all intermediate scales between these two extremes into account,} as we will see in \cref{sec:MIND}.

\subsubsection{The lasso}

The \emph{least absolute shrinkage and selection operator} (lasso; \citealp{Tib96}) has been introduced as a constraint estimator, namely the solution of 
\begin{equation*}
	\left\{
	\begin{aligned}
	&\minimize_{\beta \in \R^p}   \norm{Y- {\X} \beta}_2^2   \\
	&\st 
	\sum_{i=1}^p
	\abs{ \beta_i}
	\leq   c
	\;.
	\end{aligned}
	\right.
\end{equation*}
Note that this is the converse formulation of the constraint estimator in \eqref{eq:rss-const}. 
From this perspective, the Dantzig selector in \eqref{eq:Dantzig} is a \enquote{reverse lasso} with data $\ell^2$ fidelity term replaced by the $\ell^\infty$ norm \citep{BiRT09}. We will investigate this in more detail later on.

\subsubsection{Nemirovskii's estimator}
 \citet{Nem85} introduced for the nonparametric regression model  a particular 
 constrained  estimator which, as the lasso, in a sense 
 has  reversed roles of the constraint and the objective in \eqref{eq:rss-const}. 
 This is done in the context of Sobolev spaces $W^{k,q}$, $k \in \N, 1\le q \le \infty$. In the discrete setting of model \ref{M1} this reads as
 \begin{equation} \label{eq:nemirovski}
	\left\{
	\begin{aligned}
	&\minimize_{\beta \in \R^n}  \norm{Y-  \beta}_{{\cal N}}   \\
	&\st 
	\,\,  \norm{\beta}_{k,q} \leq  c
	\;.
	\end{aligned}
	\right.
\end{equation}
Here $\norm{\beta}_{k,q}$ is a discretized version of the $d$-dimensional $(k,q)$ Sobolev norm $\norm{f}_{W^{k,q}}:= \sum_{0 \leq |l| \leq k}|| D^l f ||_{L^q}$ {and $\beta = \samp_n \kl{f} = \kl{f \skl{x_i}}_{i=1}^n$.}  The norm $\norm{\cdot}_{\cal N}$ is a multiscale analog to the $\ell^\infty$-norm and defined as 
\begin{equation*}
\norm{\beta}_{\cal N} := \sup_{B \in {\cal N}} \frac{1}{\sqrt{|B|}} \abs{\sum_{i \in B} \beta_i}
\end{equation*}
for a system $\cal N$ of sets $B \subset \{1, \cdots, n \}$. As suggested by \citet{Nem85}, $\cal N$ is called \emph{normal} if it obeys a certain covering property, e.g.\ the system of all subsquares or discrete balls does satisfy such a condition. In fact, normal systems of cardinality $\mathcal O(n)$ exist \citep{GrLM18}.  

\subsection{Outline of this paper}

The major aim of this paper is to unify these estimators (and  extensions, which we will discuss later on) and shed some light on their commonalities and differences from a variational point of view.  This allows for some better statistical understanding but also serves as a guide for a unifying algorithmic treatment. We stress that most of this is known and scattered over the literature in different contexts and communities. {However,} we are not aware of a comprehensive and unifying  approach for all these estimators from the view point of optimization characterization. 

The rest of the paper is organized as follows. In \cref{sec:MIND}, the general class of  \enquote{MIND} (MultIscale Nemirovskii--Dantzig) estimators is introduced, which comprises all of the aforementioned estimators and generalizations thereof. In the course of  \cref{sec:threshold,sec:variational} selective instances of MIND are discussed. This includes, for example, various (block) thresholding strategies, multiscale extensions of the Dantzig selector, the group lasso, and (reverse) Nemirovskii's estimator.  The distributional properties of the multiscale constraint of MIND, are summarized in \cref{sec:dist}. The implications for the selection of a proper threshold for MIND are also discussed there.  In \cref{sec:num} {various} algorithms for the computation of MIND are discussed, based on which several numerical examples are presented as illustrations. All proofs are deferred to the Appendix.

\section{MIND: The multiscale Nemirovskii--Dantzig estimator}\label{sec:MIND}

Classical variational  methods, such as penalized or constrained least squares, control the residual vector $Y - \X \beta $ {in a global way (thus on a single scale only)}, in general. This is in contrast to dictionary based multiscale methods, like wavelet estimators. 

The MIND estimator introduced in \cref{def:mind} below can be seen as a hybrid approach between variational methods and  dictionary (e.g.\ wavelet) based multiscale methods. It follows the philosophy of dictionary based methods to analyze the residual vector simultaneously over a whole family of scales and locations. At the same time it does not necessarily rely on an explicit dictionary expansion for its regularization term. Instead it allows for a more general variational regularization formulation to employ smoothness information about the unknown $\beta$ in terms of the regularization functional $R$.

\begin{definition}[MultIscale\label{def:msd} Nemirovskii--Dantzig estimator, MIND; \citealp{GrLM18}]\label{def:mind}
For  a given index set $\KK$, let $\skl{\T_{\kk}}_{\kk\in\KK}$
be a  family of linear transformations
\begin{equation*}
    \T_{\kk} \colon \R^{n} \to \R^{n_\kk}\,, \quad  \text{ for } \kk\in\KK \,,
\end{equation*}
$\skl{w_\kk}_{\kk\in\KK}$ a family of positive numbers and $\skl{s_\kk}_{\kk\in\KK}$  a family of nonnegative  numbers. Let also $R \colon \R^p \to \R \cup \set{\infty}$ be a functional.
Any solution of the constrained optimization problem
\begin{equation} \label{eq:mind}
	\left\{
   \begin{aligned}
	&\minimize_{\beta \in \R^p}  \quad \rf\skl{\beta}   \\
	&\st \quad
	\max_{\kk\in\KK}\left \{
	\frac{\norm{ \T_{\kk} \kl{ Y - \X \beta } }_2}{w_\kk}-s_\kk  \right \}
	\leq   q
	\end{aligned}
\right.
\end{equation}
is called a \emph{MIND} for \eqref{eq:ip} with regularizer $R$, probe functionals $\skl{\T_{\kk}}_{\kk \in \KK}$, weights $\kl{w_\kk}_{\kk \in \KK}$, $\skl{s_\kk}_{\kk\in\KK}$ and threshold $q>0$.
\end{definition}

Obviously, the  Dantzig selector in \eqref{eq:Dantzig} is a special instance of~\eqref{eq:mind} as the columns of $X$ can be absorbed into the matrix $\T_{\kk}$; Also $w_\kk =1$, $s_\kk = 0$ and $\rf\skl{\beta}$ equals the (one-dimensional) discrete total variation. Many other examples will be given later. 

The weights $w_\kk$ could be absorbed into $\T_{\kk}$, but, because they often play a particular role as scale factors, we do not. Further, the penalty $s_\kk$ is included in order to balance the random perturbations caused by noise over different scales. This balancing idea was formalized by \citet{DueSpo01}, see also \citet{WaPe20} {and the references therein}. The proper choices of $w_\kk$ and  $s_\kk$ depend on the type of signal preferably to be reconstructed and the model itself, see e.g.\ \citet{SchMunDue13} {for convolution models, and} \citet{Spo09,FriMunSie14,PeSM17} and \citet{LiGM19} {for change point regression.} The constraint in \eqref{eq:mind} forces the residuals $Y - \X \beta$ to satisfy
$$	
\norm{ \T_{\kk} \kl{ Y - \X \beta } }_2^2
	=
	\sum_{i=1}^ {n_\kk}
	\abs{  \kl{ \T_{\kk}  \kl{Y-  \X \beta}}_i  }^2
	\leq ( q+ s_\kk )^2 w_\kk^2 \,,
$$
simultaneously for every $\kk \in \KK$. The threshold $q$ determines  the size of the feasibility region and acts as a tuning parameter: the larger $q$ the more the constraint in \eqref{eq:mind} is relaxed and hence the smoother (measured in terms of the functional $R$) the MIND will be. 

{The collection of probe functionals $(\T_{\kk})_{\kk \in \KK}$ encodes \emph{multiple scales}: For instance, if $\T_{\kk} v = \inner{\phi_\kk}{v}$ for some vector $\phi_\kk \in \R^n$, namely, $n_{\kk} = 1$, then the size of the support of $\phi_\kk$ (i.e.\ the number of nonzero entries of $\phi_{\kk}$) is interpreted as its \enquote{scale}.}
Probe functionals are often motivated by examining whether there is remaining structure left in the residual $Y - \X \beta$ for a candidate $\beta$. This relates to the detection of anomalies (hot spots) in (spatial) random fields, see e.g.\ \citet{ShAC16} and \citet{PrWM18} for {Gaussian errors}, and \citet{KoMW20} for extension to non-Gaussian models. Roughly, {the detection of signal is formalized as} a multiple testing problem for the collection of hypotheses
$$
H_{\kk} \T_{\kk} \kl{ Y - \X \beta } \text{ contains purely noise,} \qquad\text{ for }\; \kk\in\KK.
$$ 
In this sense, the constraint in \eqref{eq:mind} is interpreted as the acceptance region of a multiple test, and thus every MIND aims at finding the most regular candidate measured by $R$ within this acceptance region. That is, MIND can be seen as \emph{a combination of multiple testing and variational estimation.} This does not only provide guidance in designing the probe functionals {$\T_{\kk}$,} but also suggest rules for the choice of threshold $q$. For instance, a reasonable rule is to control the \emph{familywise error rate}, which is the probability of making any wrong rejections. It suggests to select $q$ such that 
\begin{equation}\label{eq:cqnt}
\inf_{\beta\in \R^p}\,\Po_\beta\set{\norm{ \T_{\kk} \kl{ Y - \X \beta } }_2^2 \leq ( q+ s_\kk )^2 w_\kk^2\quad \text{ for all }\; \kk \in \KK} \ge 1-\alpha
\end{equation}
for some {error level} $\alpha \in (0,1)$. As the value of $q$ is independent of $\beta$, \cref{def:mind} readily implies the \emph{simultaneous statistical guarantee}
\begin{equation}\label{eq:ssg}
\inf_{\beta\in \R^p}\,\Po_\beta\set{\rf\skl{\hat\beta} \le \rf\skl{\beta}} \ge 1-\alpha,
\end{equation}
where $\hat\beta$ denotes the MIND. Note that the objective $\rf$ in \eqref{eq:mind} ensures the MIND to fulfill certain desired regularity properties. Such a statistical guarantee reveals a statistical balancing between data approximation and regularity of $\beta$. The data fidelity constraint of MIND in \eqref{eq:mind} enforces the closeness to the data, but at the same hand the minimization of $\rf$ provides also smoothness control: \eqref{eq:ssg} says that, independent of the true parameter $\beta$, the MIND is no less regular (in terms of $R$) than the true parameter with probability at least $1-\alpha$. 

The quantile $q\equiv q_\alpha$ can be estimated via Monte Carlo simulations (see e.g.\ \citealp{FriMarMun12b}). Besides it can be approximated using the limiting distribution of the \emph{multiscale statistic} (recall \eqref{eq:mind})
\begin{equation}\label{eq:mstat}
T_n \;\equiv\; \max_{\kk\in\KK}\left \{
	\frac{\norm{ \T_{\kk} \kl{ Y - \X \beta } }_2}{w_\kk}-s_\kk  \right \}\; = \; \max_{\kk\in\KK}\left \{
	\frac{\norm{ \T_{\kk} \eps }_2}{w_\kk}-s_\kk  \right \},
\end{equation}
which will be discussed in \cref{sec:dist}. Such choices of $q$ turn out to be extremely favorable in practice, since we do not require any knowledge of the true parameter, or to estimate it from the data, in contrast to the penalized or the constrained optimization formulation with switching roles of objective and constraint.

In the sequel we consider several further estimation and thresholding techniques that are shown to be particular examples of  the MIND.

\section{Thresholding methods}
\label{sec:threshold}

In this section we {relate the MIND principle to so-called} thresholding based estimators for \eqref{eq:ip} which rely on an (explicit) expansion of $\beta$. Most thresholding techniques have been initially designed for solving
the nonparametric regression problem \ref{M1}. An extension to the linear inverse problem is the wavelet--vaguelette decomposition,
which will be discussed in \cref{sec:wave-vag} in more detail.

\subsection{Wavelet soft-thresholding}
\label{ss:wst}

Let $\kl{ \phi_\la}_{\la \in \La}$ be an orthonormal {(discrete)} wavelet
basis of $\R^n$, i.e.\ $\norm{\phi_\la}_2 =  1$ and $\inner{\phi_\la}{ \phi_{\la'}} = 0$ for $\la \neq \la'$,
see for example \citet{Vid99} and \citet{Mal09}.
Due to orthonormality, every parameter vector $\beta  \in \R^n$ can be uniquely
expanded in the wavelet basis
\begin{equation*}
\beta  = \sum_{\la \in \La} \inner{\phi_\la}{\beta} \phi_\la \,.
\end{equation*}

\begin{example}[Haar wavelets]\label{ex:haar}{
The Haar father wavelet (or scaling function) $\phi^{(\textsc{h})}$ is defined as 
$$
\phi^{(\textsc{h})}(x) = \begin{cases}
1\quad &\text{if }\; 0 \le x <1,\\
0 \quad & \text{otherwise,}
\end{cases}
$$
and the Haar mother wavelet (or wavelet function) as $\psi^{(\textsc{h})}(x) \coloneqq \phi^{(\textsc{h})}(2x) - \phi^{(\textsc{h})}(2x-1)$ for $x \in \R$.  Let $\psi^{(\textsc{h})}_{j,k} (x)= 2^{j/2}\psi^{(\textsc{h})}(2^j x- k)$ for $x \in \R$, where $j\in \N_0$ corresponds to the \emph{scale} and $k\in \N_0$ to the \emph{location} of $\psi^{(\textsc{h})}_{j,k}$. For simplicity, we consider  $n = 2^J$ with some $J \in \N_0$. It is known that 
$$
\mathcal{B} \equiv \set{\phi^{(\textsc{h},n)},\, \psi^{(\textsc{h},n)}_{j,k}\;:\; j = 0, \ldots, J-1 \text{ and } k =0,\ldots, 2^j-1}
$$ 
forms an orthonormal basis of $\R^n$ (see e.g.\ \citealp{Mal09}), where the included basis elements are defined by sampling the scaling function and the scaled wavelet functions, 
$$
\phi^{(\textsc{h},n)} =\samp_n\bigl( \phi^{(\textsc{h})}\bigr) \coloneqq \bigl(\phi^{(\textsc{h})}(i/n)\bigr)_{i=0}^{n-1}\quad \text{and}\quad \phi^{(\textsc{h},n)}_{j,k} = \samp_n \bigl(\psi^{(\textsc{h})}_{j,k}\bigr) \coloneqq \bigl(\psi^{(\textsc{h})}_{j,k}(i/n)\bigr)_{i=0}^{n-1}.
$$ 
In our notation, the basis $\mathcal{B}$ will be re-indexed as $\sset{\phi_\la: \la \in \La}$.}

{
Some other examples of orthonormal wavelets include Daubechies’ wavelets, Coiflets, Meyer wavelets and B-spline wavelets, see e.g.\ \citet{Chu92} and \citet{Dau92}.}
\end{example}

For any  $q > 0$, the
nonlinear \emph{soft-thresholding function}  is given as
$\eta^{\mathrm{(soft)}}\kl{\,\cdot\,, q}
\colon
\R \to \R
$,
\begin{equation*}
	\eta^{\mathrm{(soft)}}\kl{x, q}
	\coloneqq
	 x\kl{1 - \frac{q}{\sabs{x}}}_+  \; \text{ for } x \in \R\,.
\end{equation*}
The soft-thresholding function sets any coefficient $x$ with $\sabs{x}$ below  the threshold $q$ to zero and shrinks the remaining coefficients towards zero by the constant value $q$. Moreover, one notices that $\eta^{\mathrm{(soft)}}\kl{\,\cdot\,, q}$ is a continuous function on the whole real line, see \figref{fig:thresholdingfunctions}.

This leads to one of the most prominent wavelet shrinkage methods, the \emph{wavelet soft-thresholding estimator} \citep{Don95}
\begin{equation} \label{eq:st}
\hat \beta^{\mathrm{(soft)}}
=
\sum_{\la \in \La}
\eta^{\mathrm{(soft)}} \kl{ \inner{\phi_\la}{Y},  q_\la }
\phi_\la \,.
\end{equation}
Here the thresholds $q_\la > 0 $, for $\la \in \La$,
are tuning parameters that determine whether or not an empirical
wavelet coefficient $\inner{\phi_\la}{Y}$ is accepted as a wavelet coefficient
for the estimate $\hat \beta^{\mathrm{(soft)}}$.  In its elementary form, wavelet soft-thresholding
is often used with the \emph{universal threshold}  $q_\la = q \coloneqq
\sigma \sqrt{2 \log n} $,  which is independent of the index  $\la \in \La$, see \cite{Don95} and \cref{sec:dist}.
Extensions to scale or data dependent thresholds have been investigated in \citet{DonJoh95} and \citet{AntFan01} among many others.

Using the fast wavelet transform algorithm (see e.g\ \citealp{Dau92,Coh03,Mal09}) the wavelet thresholding estimator can be computed with only $\mathcal O  \kl{n}$ floating point operations. Despite the simplicity and computational efficiency, soft-thresholding  has various statistical optimal recovery  properties, such as adaptive minimax optimality over various Besov balls \citep{DonJoh94,DonJoh95}. 

The next result reveals soft-thresholding as an instance of the MIND in \eqref{eq:mind} {with probe functionals $\T_{\kk} = \T_{\la}$ defined to map $v\in \R^n$ to $\inner{\phi_\la}{v}$, and weights taken as $w_\kk = w_{\la} = q_\la$ and $s_\kk = 0$}.

\begin{theorem}[Variational characterization of soft-thresholding; \citealp{Don95}]\label{thm:soft}
Let $\kl{ \phi_\la}_{\la \in \La}$ be an orthonormal basis of $\R^n$.
For any $r \in (0, \infty)$,
the soft-thresholding estimator $\hat \beta^{\mathrm{(soft)}}$ defined by \eqref{eq:st} is the unique solution of
 \begin{equation*} 
	\left\{
	\begin{aligned}
	&\minimize_{\beta \in \R^n}  &&   \sum_{\la\in \La} \abs{\inner{\phi_\la}{\beta}}^r   \\
	&\st &&
	\max_{\la \in \La }
	\frac{\abs{\inner{\phi_\la}{  Y - \beta }}}{q_\la}
	\leq   1
	\,.
	\end{aligned}
	\right.
\end{equation*}
The same holds true, if we replace the objective function above by $\norm{\beta}_2^2  = \sum_{i=1}^n \beta_i^2$.
\end{theorem}

Note that the variational characterization of soft-thresholding  in \cref{thm:soft}
holds for any orthonormal basis $\kl{ \phi_\la}_{\la \in \La}$, e.g.\ a Fourier basis in place of
the wavelet basis. However, soft-thresholding is often applied with a wavelet basis
(or a similar multiscale system) since its optimality properties for function estimation heavily
depends on the multiscale structure and the spatial adaptivity  of wavelets \citep{Don95}.

Wavelet soft-thresholding can further be equivalently characterized as the unique solution of the penalized least squares functional (cf.\ \eqref{eq:rss-pen})
\begin{equation*} 
	\minimize_{\beta \in \R^p}  \frac{1}{2} \norm{ Y - \beta }_2^2
	+
     \sum_{\la \in \La} q_\la \abs{ \inner{\phi_\la}{\beta} } \,,
\end{equation*}
where the penalty $R$ is taken as the  $\ell^1$-norm of  the wavelet coefficients of the parameter vector (see \citealp{Don95,ChaDeVLeeLuc98}). This characterization of  soft-thresholding again can be verified in an elementary manner following  the proof of \cref{thm:soft}. Here the equivalence of the penalized and the constrained problem is, however,  not accidental. In fact, this equivalence is a special case of  \cref{thm:group-lasso} below and due to Fenchel duality (see Appendix \ref{app:group-lasso}).

\subsection{Modified thresholding nonlinearity}

The soft-thresholding function $\eta^{\mathrm{(soft)}}\kl{\,\cdot\,, q}$ systematically shrinks coefficients towards zero, at least if the thresholds are taken to be independent of the data. This yields a finite sample bias  if there is a significant amount of non-zero coefficients. To overcome this issue, various modified thresholding functions have been proposed in the literature.

\begin{itemize}
\item
\textsc{Hard-thresholding function \citep{DonJoh94}}.\\
 Besides the soft-thresholding function, the \emph{hard-thresholding function},
\begin{equation*} 
	\eta^{\mathrm{(hard)}} \kl{x, q}
	\coloneqq
	\begin{cases}
	0  & \text{ if }  \abs x \leq  q\\
	x  &  \text{ otherwise }	 \,,
	 \end{cases}
\end{equation*}
is the most basic and best known  thresholding function.
The hard-thresholding function does not shrink large coefficients and therefore
usually yields a smaller mean square error than  soft-thresholding.
However, the hard-thresholding function has discontinuities at $x = \pm q$,
which sometimes yields visually disturbing artifacts for signal and image recovery.

\item
\textsc{Nonnegative garrote \citep{Bre95,Gao98}}.\\
The  \emph{nonnegative garrote thresholding function} is a one-dimensional
version of the  well known James--Stein shrinkage function \citep{JamSte61},
and is defined by
\begin{equation*}
	\eta^{\mathrm{(JS)}} \kl{x, q}
	\coloneqq
	x  \kl{1 - \frac{q^2}{\sabs{x}^2}}_+  \,.
\end{equation*}
It is continuous and has a vanishing shrinkage effect as $\sabs{x}$
tends to infinity. Therefore, it is often claimed to  combine the advantages of soft- and
hard-thresholding.
\end{itemize}

\begin{figure}
\centering
\includegraphics[width=0.6\textwidth]{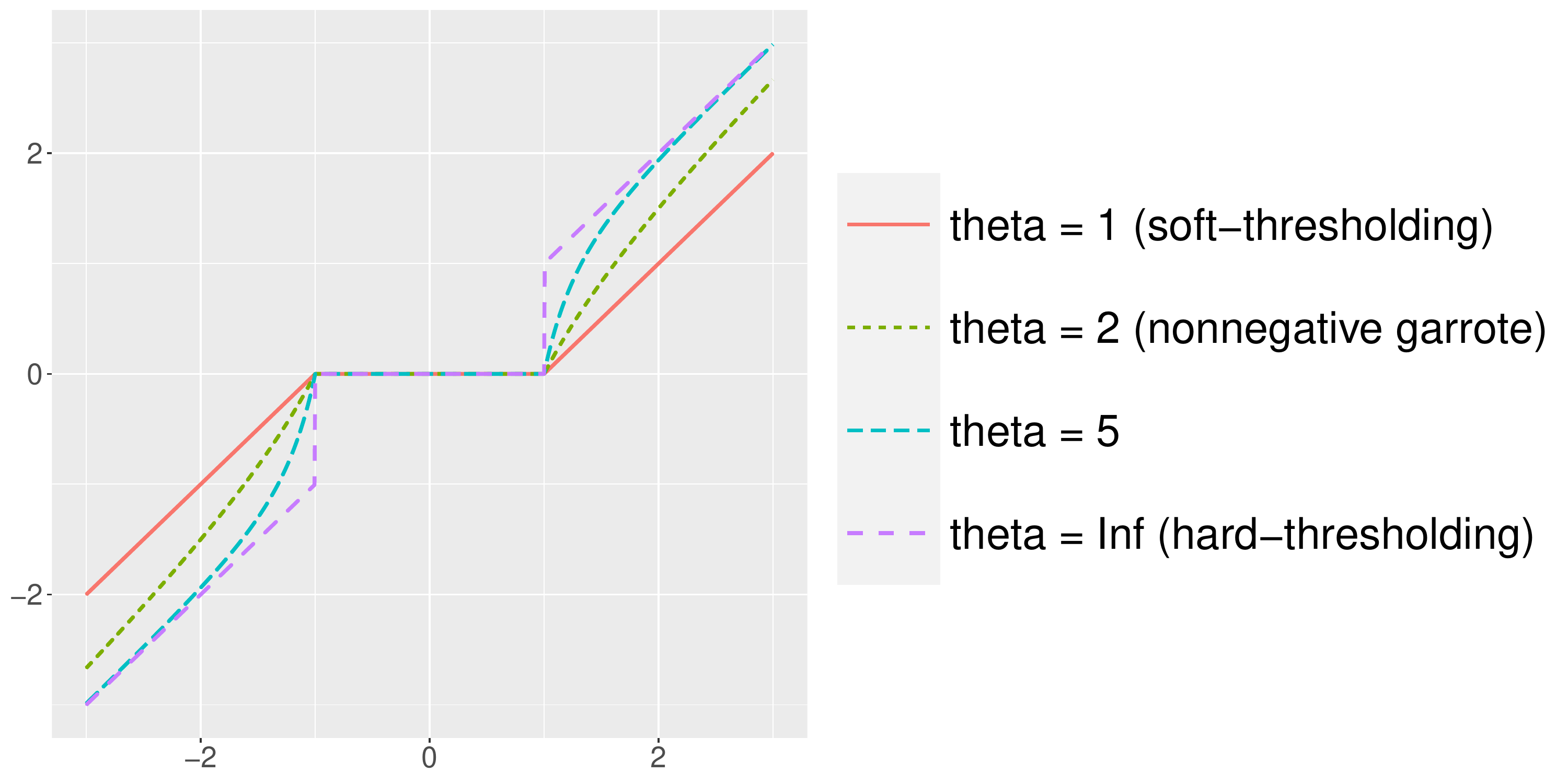}
\caption{The thresholding functions $\eta_\theta \kl{\,\cdot\,, q}$, plotted for the threshold $q=1$, are antisymmetric and equal to zero on $\set{x : \abs{x} \leq q}$. It coincides with the soft-thresholding  for $\theta=1$, with the nonnegative garrote for $\theta=2$, and converges pointwise and monotonically to the hard-thresholding for $\theta \to \infty$.}
\label{fig:thresholdingfunctions}
\end{figure}

All thresholding functions defined above are special cases of the class of functions
\begin{equation*} 
	\eta_\theta \kl{x, q}
	\coloneqq
	x  \kl{1 - \frac{q^\theta}{\sabs{x}^\theta}}_+
\end{equation*}
for some specific value of $\theta >0$. In fact, the soft-thresholding function corresponds to $\theta=1$, the nonnegative garrote to $\theta=2$, and the hard-thresholding function to the limiting case $\theta \to \infty$.  See again \figref{fig:thresholdingfunctions}. 

We observe that the soft-thresholding function and the nonnegative garrote shrinkage function are related via the explicit relation
\begin{equation*}
	\eta^{\mathrm{(JS)}} \kl{x,q}
	=
	\eta^{\mathrm{(soft)}} \kl{x, \frac{q^2}{\sabs{x}} }
	\, \text{ for }  x \neq 0  \,.
\end{equation*}
(A similar relation, of course, holds for any of the thresholding functions $\eta_\theta$.) This basic identity allows to interpret shrinkage by the  nonnegative garrote as soft-thresholding  applied with the threshold $q^2/\sabs{x}$, which is now data dependent. Based on this  simple observation, one can carry over many properties of the soft-thresholding estimator to the \emph{nonnegative garrote {(or James--Stein)} estimator}
\begin{equation*}
\hat \beta^{\mathrm{(JS)}}
=
\sum_{\la \in \La}
\eta^{\mathrm{(JS)}} \kl{ \inner{\phi_\la}{Y},  q_\la }
\phi_\la
 \,.
\end{equation*}

\begin{theorem}[Variational characterization of the nonnegative garrote]
Assume the setting of \cref{thm:soft}.
For\label{thm:nng} any $r \in (0, \infty)$, the nonnegative garrote estimator $\hat \beta^{\mathrm{(JS)}}$ defined above is  the unique solution of
 \begin{equation*}
	\left\{
	\begin{aligned}
	&\minimize_{\beta \in \R^p} &&   \sum_\la \abs{\inner{\phi_\la}{\beta}}^r   \\
	&\st &&
	\max_{\la \in \La}
	\frac{\abs{\inner{\phi_\la}{  Y - \beta }}}{w_\la}
	\leq   1
	\,.
	\end{aligned}
	\right.
\end{equation*}
Here $w_\la = {q_\la^2}/{\max\set{q_\la, \abs{\inner{\phi_\la}{ Y }}}} $ are weights depending on the data coefficients $\inner{\phi_\la}{ Y }$ and the thresholds $q_\la$. The same holds true, if we replace the objective function by $\norm{\beta}^2_2$.
\end{theorem}

The optimization problem in \cref{thm:nng} is obviously an instance of the MIND in \eqref{eq:mind}, with objective $\rf\skl{\beta} = \sum_\la \abs{\inner{\phi_\la}{\beta}}^r$, probe functionals $\kl{\inner{\phi_\la}{ \, \cdot \, } \colon \la \in \La}$, and data dependent weights $\kl{w_\la \colon \la \in \La}$. Allowing data dependent weights reveal almost any thresholding technique  as a MIND. Even the hard-thresholding estimator may be written as a solution of  the minimization problem in \cref{thm:nng} if one replaces the weights  by
\begin{equation*}
w_\la \equiv w_\la  \kl{\inner{\phi_\la}{ Y }, q_\la} =
\begin{cases}
q_\la  & \text{ if } \abs{\inner{\phi_\la}{ Y }} \leq q_\la \\
0  & \text{ if } \abs{\inner{\phi_\la}{ Y }}  >  q_\la  \,.
\end{cases}
\end{equation*}
Here we set $0/0=1$  and $x/0=\infty$ for $x >0$. In some sense,  hard-thresholding is a degenerate situation, where the reciprocal weights $1/w_\la$ become singular for $\abs{\inner{\phi_\la}{ Y }}  >  q_\la$.

\begin{figure}[thb]
\centering
\includegraphics[width=0.8\textwidth]{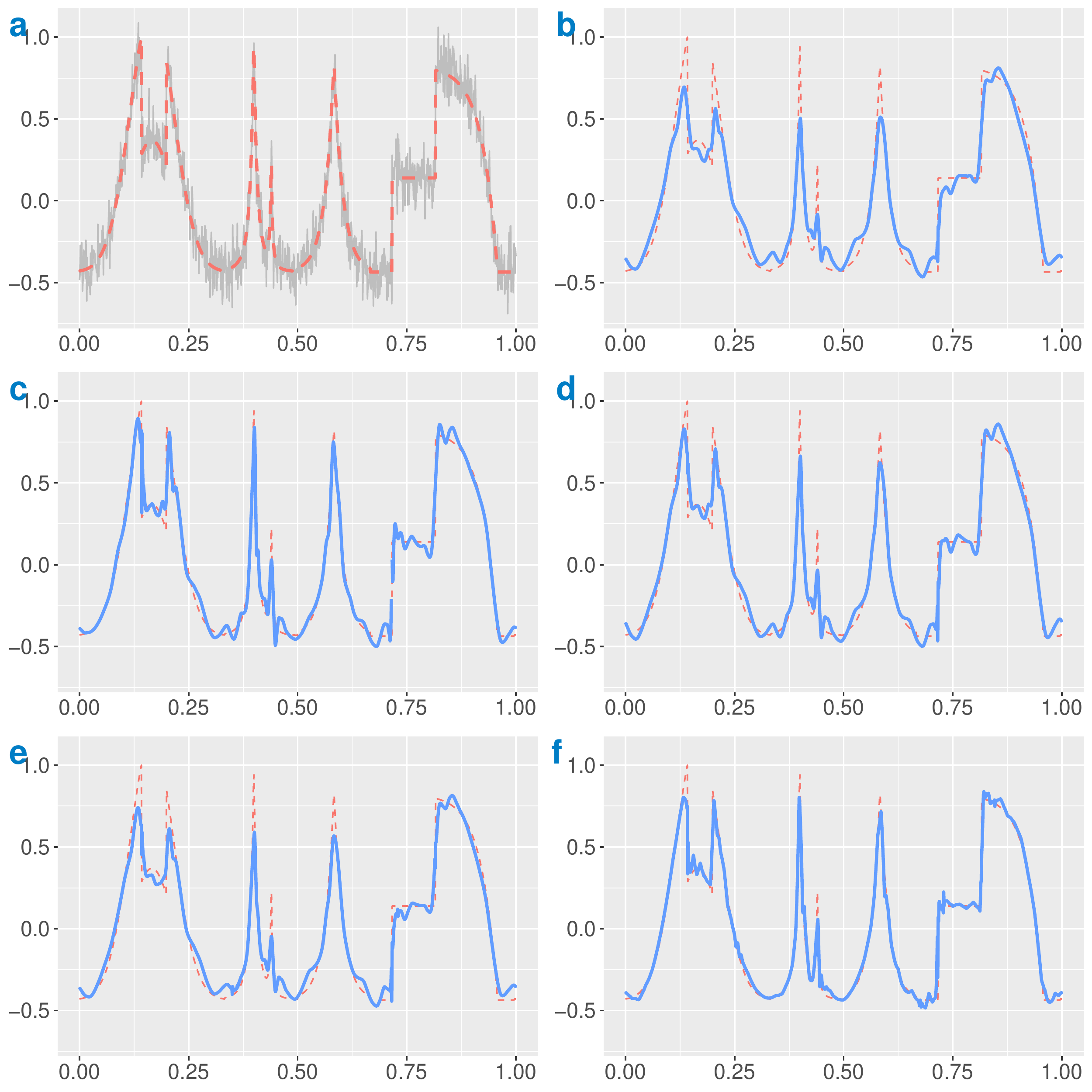}
 \caption{Regression of \enquote{piecewise smooth} signal via various thresholding strategies. (a) Noisy data (solid gray line) of 1,024 samples in model \ref{M1} with $\mathcal{N}(0,0.1^2)$ error; (b)--(f) Estimators (solid blue line) of soft-thresholding, hard-thresholding, nonnegative garrote, FDR soft-thresholding \citep{AbrBen96} and block soft-thresholding, respectively. The true signal (dashed red line) is plotted in all panels. Daubechies' least asymmetric (symlets) with six vanishing moments are used.}
 \label{fig:thresholding}
 \end{figure}

We illustrate the performance of various thresholding methods in \figref{fig:thresholding}. It shows that soft-thresholding shrinks sharp peaks while hard-thresholding introduces oscillating artifacts; Nonnegative garrote is a compromise between the two. Besides the choice of {the nonlinear shrinkage function}, the performance can be improved e.g.\ by an FDR (false discovery rate) strategy  \citep{AbrBen96}, or by a block-wise strategy (discussed below). 
 
\subsection{Block thresholding}
\label{sec:blockthresholding}

In wavelet block thresholding the thresholding operation is not applied separately to each  individual wavelet coefficient, but uniformly for a whole group of  wavelet coefficients (e.g.\  \citealp{HalEtal97,Cai99,CaiZho09,CheFadSta10}).
For that purpose, the set of all wavelet indices
\begin{equation*}
\La  =  \bigcup_{\kk \in \KK}  \La_\kk
\end{equation*}
is  grouped into several \emph{disjoint} subsets $\La_\kk \subset \La$.
One defines, for any $\kk \in \KK$, blocks of wavelet coefficients
\begin{equation}\label{eq:wave-block}
	\base_\kk	Y
	\coloneqq
	\kl{\inner{\phi_\la}{Y} : \la \in \La_\kk} \,.
\end{equation}
The wavelet thresholding is then performed uniformly for all coefficients within the same block. As in the case of component-wise thresholding, every particular block thresholding estimator depends on the thresholding function. 
A class of block thresholding function is  
$$
\eta_\theta(x, q) :=  x \kl{1 -\frac{q^\theta}{\norm{x}_2^\theta} }_+\qquad \text{for }\quad \theta > 0,
$$
which now is applied to the whole  block of coefficients,
$ x = \skl{ x_\la  : \la \in \La_\kk} \in \R^{\La_\kk}$, instead of a single coefficient. 
The \emph{block soft-thresholding estimator} and the \emph{block James--Stein estimator}, respectively, are then defined by 
\begin{align} 
	\hat \beta^{\text{(b-soft)}}
	& =  
	\sum_{\kk \in \KK}
	\sum_{\la \in \La_\kk}
	\kl{ \eta_1 \kl{\base_\kk	Y , q_\kk} }_\la
	\phi_\la\,,\label{eq:block-st2}
	\\
	\hat \beta^{\text{(b-JS)}}
	& = 
	\sum_{\kk \in \KK}
	\sum_{\la \in \La_\kk}
	\kl{ \eta_2 \kl{\base_\kk	Y ,q_\kk} }_\la
	\phi_\la
	\,,\nonumber
\end{align}
where $q_\kk >0$ are possibly block dependent thresholds. Similarly to \cref{thm:soft}, the block soft-thresholding estimator and the block James--Stein estimator are again special instances of the MIND in~\eqref{eq:mind}.

\begin{theorem}[Variational characterization of block soft-thresholding]
For\label{thm:soft-block} any $r \in (0, \infty)$, the block soft-thresholding estimator $\hat \beta^{\text{(b-soft)}}$ in \eqref{eq:block-st2} is  the unique solution of
 \begin{equation*} 
	\left\{
	\begin{aligned}
	&\minimize_{\beta \in \R^p}  &&   \sum_\la \abs{\inner{\phi_\la}{\beta}}^r   \\
	&\st &&
	\max_{\kk \in \KK}
	\frac{\norm{\base_\kk\kl{  Y - \beta }}_\infty}{q_\kk}
	\leq   1
	\,.
	\end{aligned}
	\right.
\end{equation*}
The same holds true if we replace the objective function above  by $\norm{\beta}_2^2 = \sum_i \beta_i^2$.
\end{theorem}
As in the non-block case,  the block James--Stein thresholding $\eta^{\mathrm{(b-JS)}} \kl{x,q}$ can be expressed in terms of block soft-thresholding applied with the threshold $q^2/\norm{x}_2$. Hence the block James--Stein estimator can be characterized by the minimization in \cref{thm:soft-block} with $q_\kk$ replaced by
\begin{equation*}
w_\kk= w_\kk\skl{\base_\kk Y , q_\kk} =
 \frac{q_\kk^2}{\max\sset{q_\kk, \snorm{\base_\kk Y}_2}} \,.
\end{equation*}
The same again holds true if we replace  the objective by  the squared $\ell^2$ norm
$\norm{\beta}_2^2$.

{Recall that a wavelet basis is indexed by scales and locations, see \cref{ex:haar}.}
In wavelet block thresholding, each set $\La_\kk$ is  usually supposed to
have a fixed scale and to  consist of an interval (or block)  of  location indices.  We emphasize
that for the  above characterizations of  the block thresholding methods
such an additional restriction is not required. Further, the wavelet basis
may be replaced by an arbitrary orthonormal basis.

\subsection{Thresholding in  frames}

The rationale behind  wavelet thresholding is that the parameter to be recovered can be efficiently represented as a sparse linear combination of elements of the wavelet basis. In {real world} signal and image processing applications the parameter $\beta$ is usually not strictly sparse and the removal of small coefficients often introduces visually disturbing artifact (see \figref{fig:thresholding} and e.g.\ \citealp{DonJoh94,CoiDon95,Don95,CanDon99,Mal09,StaMurFad10,GrLM18}).

One way of addressing this issue is to consider an overcomplete frame or dictionary instead of a wavelet basis. A \emph{dictionary} of $\R^n$ is {a} family $\kl{ \phi_\la}_{\la \in \La}$ of elements that spans the whole space $\R^n$. Hence any $\beta \in \R^n$  can be written in the form
$\beta = \sum_{\la \in \La}  x_\la \varphi_\la$ for certain coefficients $x_\la \in \R$.
If there are some constants $0 < a \leq  b < \infty$ such that
\begin{equation}\label{eq:frmbnd}
a \norm{\beta}_2^2 \leq
	\sum_{\la \in \La}
	\abs{\inner{\phi_\la}{\beta}}^2
	\leq b \norm{\beta}_2^2
	\quad \text{for all }\beta \in \R^n \,,
\end{equation}
then $\kl{\phi_\la}_{\la \in \La}$ is called a \emph{frame.}
In such a situation, the mapping
\begin{equation*}
\base \colon \R^n \to \R^\La \colon \beta \mapsto \kl{\inner{\phi_\la}{\beta}}_{\la \in \La}
\end{equation*}
is called the \emph{analysis operator} and $\base^{\mathsf{T}} \base$ the \emph{frame operator.} Due to the  frame property, the frame operator  is invertible and the elements $\psi_\la  =   (\base^{\mathsf{T}} \base)^{-1}  \phi_\la$ are well defined. They again form a frame  $\kl{\psi_\la}_{\la \in \La}$, which is called the \emph{dual frame.} Note that in a finite dimensional situation any dictionary is automatically a frame, but this is not the case for highly redundant frames in infinite
dimensional spaces.

If $\kl{\phi_\la}_{\la \in \La}$  is a frame, then one has the reproducing  formula $\beta = \sum_{\la \in \La} \inner{\phi_\la}{\beta} \psi_\la$ where $\kl{\psi_\la}_{\la \in \La}$ its  frame dual to $\kl{\phi_\la}_{\la \in \La}$. 
This motivates the following definition of a \emph{frame based soft-thresholding estimator}
\begin{equation*}
\hat \beta^{\mathrm{(soft)}}
=
\sum_{\la \in \La}
\eta^{\mathrm{(soft)}} \kl{ \inner{\phi_\la}{Y},  q_\la }
\psi_\la
\,.
\end{equation*}
It can be written as
$\hat \beta^{\mathrm{(soft)}} =
\sum_{\la \in \La}
\hat \x_\la
\psi_\la$
with $\kl{\hat \x_\la}_{\la\in \La}$ being the unique minimizer of
\begin{equation*}
	\left\{
	\begin{aligned}
	&\min_{\x\in \R^\La}  &&   \sum_\la \abs{\x_\la}^r  \\
	&\st &&
	\max_{\la \in \La }
	\frac{\abs{\inner{\phi_\la}{  Y } - \x_\la }}{q_\la}
	\leq   1
	\end{aligned}
	\right.
	\qquad \qquad\text{for } 0 < r < \infty.
\end{equation*}
Hence the coefficients $\hat \x_\la$ can be viewed as a special case of the MIND in \eqref{eq:mind}.

Popular redundant system used for thresholding in signal  and image
processing  are translation invariant wavelet systems \citep{CoiDon95,NasSil95,LanGuoOdeBurWel96,PesKriCar96}, curvelets \citep{CanDon00,StaCanDon02,MaPl10}, shearlets \citep{LabLimKutWei05,KutMorZhu12}, contourlets \citep{DoVet05}, or needlets \citep{KKLPP10}. Its particular choice will depend on prior information of the signal, the underlying geometry of the domain or computational aspects. In general, all thresholding techniques from before (e.g.\ hard- or block thresholding) can be applied here, but as the coefficients now become dependent, there are some subtle differences to orthogonal systems, see also \cref{sec:dist}.

\subsection{Wavelet--vaguelette decomposition and related approaches}
\label{sec:wave-vag}

The wavelet--vaguelette decomposition is introduced in~\cite{Don95b} to generalize wavelet techniques from nonparametric regression to linear inverse problems. Recall model~\ref{M2}. Let $\A \colon \Fs \to \U$ be a bounded linear operator mapping the function space $\Fs \subseteq L^2(\Om)$ to a Hilbert space~$\U$.

\begin{definition}[Wavelet--vaguelette decomposition, \citealp{Don95b}]
The\label{def:wvd} family $\skl{\phi_\la, u_\la, v_\la, \kappa_\la }_{\la \in \La}$ is called a \emph{wavelet--vaguelette decomposition} for the linear operator $\A \colon \Fs \to \U$, if the following \emph{quasi-singular value decompositions} hold:
\begin{align*}
	\A \phi_\la
	&=
	\kappa_\la u_\la \,,
	\quad \text{ for }  \la   \in \La
	\,,
	\\
	\A^* v_\la
	&=
	\kappa_\la \phi_\la \,,
	\quad \text{ for }  \la   \in \La \,,
\end{align*}
where $\skl{\phi_\la}_{\la \in \La}$ is an orthonormal wavelet basis of $\Fs$, $\skl{u_\la}_{\la \in \La}$ and $\skl{v_\la}_{\la \in \La}$ are two bases of the space $\mathcal U$ such that $\sinner{v_\la}{u_{\la'}} = \delta_{\la, \la'}$ for all $\la, \la' \in \La$, and $\skl{\kappa_\la}_{\la \in \La}$ is a family of nonnegative numbers, referred to as  \emph{quasi-singular values}, which only depend the scale index but not the spatial index of the wavelets $\phi_\la$.
\end{definition}

If $\kl{\phi_\la, u_\la, v_\la, \kappa_\la }_{\la \in \La}$
is a wavelet--vaguelette decomposition for the operator $\A$,
then we have the reproducing formula
\begin{equation*}
f  =
\sum_{\la \in \La}
\frac{\inner{v_\la}{ \A f}}{\kappa_\la}  \, \phi_\la
\quad \text{ for } f  \in \Fs \,.
\end{equation*}
In the case of noisy observations $g = \A f  + \eps$,  with $\eps$ denoting a white noise process,
the \emph{wavelet--vaguelette soft-thresholding estimator} \citep{Don95b} for $f$  is defined by
\begin{equation*} 
\hat f^{\rm (WV)}
=
\sum_{\la \in \La}
\eta^{\mathrm{(soft)}}
\kl{\frac{\inner{v_\la}{g}}{\kappa_\la}, q_\la}
\, \phi_\la
\,,
\end{equation*}
with $q_\la > 0$ denoting certain scale dependent thresholds.  

\begin{theorem}[Variational formulation of the wavelet--vaguelette estimator]
For\label{thm:wvd} any $r>0$, the wavelet--vaguelette soft-thresholding estimator $\hat f^{\rm (WV)}$ is the
unique solution of
\begin{equation*} 
	\left\{
	\begin{aligned}
	&\minimize_{f\in \Fs}  &&   \sum_{\la\in \La}
	\abs{\inner{\phi_\la}{f}}^r   \\
	&\st &&
	\max_{\la \in \La }
	\frac{\abs{\inner{v_\la}{  g -  \A f }}}
	{\kappa_\la q_\la}
	\leq   1
	\,.
	\end{aligned}
	\right.
\end{equation*}
The same holds true, if we replace the objective function above by $\norm{f}_{L^2\kl{\Om}}^2$.
\end{theorem}

Reasonable ways of adjusting $\hat f^{\rm (WV)}$ to the case of discretely sampled data  $Y_ i = \kl{\samp_n \circ \A \circ \iop_p \beta}_i   + \eps_i$, with $\beta =  \samp_p \kl{f} = \kl{f \skl{x_j}}_{j=1}^p$, as in model \ref{M2} are discussed in \citet[Section~6.3]{Don95b}. {As a consequence, $\hat f^{\rm (WV)}$ can be viewed as a MIND estimator in \eqref{eq:mind}.} We omit details and restrict our representation to the continuous model  {(unlike the rest of the paper)} for simplicity.

In \cite{Don95b} wavelet--vaguelette decompositions have been derived for integration, fractional integration  and the Radon transform. Related techniques for the solution of statistical inverse problems can be found in \citet{AbrSil98}, \citet{CanDon02}, and \citet{KalMal03}. In \citet{AbrSil98} the roles of wavelets and vaguelettes are reversed.  In \citet{KalMal03} special mirror wavelet bases are constructed for certain deconvolution problems. The use of curvelets instead of wavelets was studied in \cite{CanDon02} and applied to the noisy Radon inversion.

We stress that the existence of a wavelet--vaguelette decomposition requires a certain scale invariance of the operator $\A$, which are not satisfied in general, see \cite{Don95b}.

\section{Variational methods for estimation}\label{sec:variational}

In the last section we studied thresholding approaches using explicit expansions with respect to (multiscale) systems of functions (dictionaries). We derived  equivalent variational formulations as constrained optimization problems with a multiscale constraint. The thresholding based multiscale approaches can be combined with variational regularization where no explicit dictionary is given. Again, they are casted as instances of the MIND in \eqref{eq:mind}.

In this section we provide a selective overview of such variational estimation schemes. We start with standard penalized least squares and then present hybrid approaches  combining variational and multiresolution schemes.

\subsection{Penalized least squares}\label{ss:pls}

Presumably the most basic and prominent variational estimation techniques is \emph{penalized  least squares}
\begin{equation*}
   \minimize_{\beta \in \R^p} \frac{1}{2}
    \norm{ Y - \X \beta   }_2^2
    +
    \gamma
    \rf\skl{\beta}
	 \,.
\end{equation*}
Here $R \colon \R^p \to \R \cup \set{\infty}$
is some regularization functional and $\gamma >0$ a penalty parameter. In the special case $ R \kl{\beta} = \norm{ \beta }_2^2$ the penalized least squares is known as \emph{ridge regression}.

By interpreting $1/\gamma > 0 $ as Lagrange multiplier, the penalized least squares estimator  can be written in the constrained form
\begin{equation*}
	\left\{
	\begin{aligned}
	&\minimize_{\beta \in \R^p}  &&  \rf\skl{\beta}   \\
	&\st &&
	\norm{ Y - \X \beta   }_2
	\leq   q
	\;.
	\end{aligned}
	\right.
\end{equation*}
In fact, the constrained optimization and its unconstrained version are essentially  equivalent if the Lagrangian parameter
is chosen according to Morozov's discrepancy principle (see \cref{thm:lagrange} in the Appendix for a precise statement). The constrained optimization problem is obviously a particular case of the MIND in \eqref{eq:mind} if one considers the identity $\Id \colon \R^n \to \R^n$ as the only probe functional. {In this sense, it has only a single scale.}

{Another well studied instance of regularized least squares are the lasso and its variants that will be studied in \cref{sec:lasso}.}

\subsection{{Total variation regularization}}\label{ss:hdtv}
Thresholding based multiscale methods provide spatial adaptivity and are known to be optimal for function estimation in Besov balls. However, they often show visually disturbing artifacts due to missing band pass information. In contrast, variational methods, such as total variation regularization, often produce visually more appealing results. 

Choosing the one-dimensional discrete total variation as regularization functional, the constrained least squares estimator takes the form
\begin{equation} \label{eq:tv-const}
	\left\{
	\begin{aligned}
	&\minimize_{\beta \in \R^p}  &&   \sum_{j=1}^{p-1} \abs{ \beta_{j+1} - \beta_j} \\
	&\st &&
	\norm{  Y - \X \beta   }_2
	\leq   q
	\;.
	\end{aligned}
	\right.
\end{equation}
This is \emph{discrete} total variation regularization, also known as \emph{trend filtering} \citep{Tib14}. In the one-dimensional nonparametric regression case (model~\ref{M1}), a variant with \emph{continuous} total variation as regularization instead is studied and shown to be minimax optimal over  bounded variation function classes in \citet{MamGee97}; See also \citet{DavKov01} and \citet{DavKovMei09}. In the case that $\beta$ represents an image, the two-dimensional analogs and the penalized version are instances of the Rudin--Osher--Fatemi (ROF) denoising model introduced in~\cite{RudOshFat92}. There are extensions to higher orders of total variation and/or general dimensions \citep{HuRi16,SWST17,FaGS19,GLCS20,OvdG20}, to tensors \citep{OvdG21}, and also to graphs \citep{WSST16}. See also~\cite{ChaLio97}, where explicit relations between constraint and unconstraint total variation minimization have been derived in a very general (infinite dimensional) setting. {All these approaches regularize the total variation functional in a global \enquote{monoscale} fashion, as in \eqref{eq:tv-const}. In the following, we will discuss its multiscale extension to systems of scales.}

{This is motivated from empirical studies which show that monoscale total variation methods} seem to lack spatial adaptivity to varying smoothness of the underlying signal or image. See e.g.\ the discussion in \citet[Section~5.1]{CanGuo02}, and also \cref{sec:num}.

Hence, one wishes to combine \enquote{the best of both worlds}, by solving
\begin{equation} \label{eq:combined}
	\left\{
	\begin{aligned}
	&\minimize_{\beta \in \R^p}  &&   \norm{\beta}_{\mathrm TV}   \\
	&\st &&
	\max_{\la \in \La }
	\frac{\abs{\inner{\phi_\la}{  Y - \X \beta }}}{q_\la}
	\leq   1
	\,,
	\end{aligned}
	\right.
\end{equation}
where $\norm{\cdot}_{\mathrm TV}$ represents the total variation, and $\kl{\phi_\la}_{\la \in \La}$ is a multiscale system such as a wavelet basis or a shearlet frame.

The hybrid approach of \eqref{eq:combined} has been introduced independently by several authors. 
In  \cite{CanGuo02} and \citet{StaDonCan01} the formulation in \eqref{eq:combined} is proposed for $\X = \Id$ in combination with overcomplete dictionaries. \citet{ChaZho00} and \cite{DurFro01} used a wavelet basis. Finally,  \cite{Mal02b,Mal02a} studied \eqref{eq:combined} for general $\X$. In \citet{ALM18}, the hybrid {total variation estimator  for $\X = \Id$} in \eqref{eq:combined} with various dictionaries, e.g.\ wavelets, curvelets or shearlets, is shown to be asymptotically minimax optimal (up to a logarithmic factor) with respect to $L^q$-risk ($1\le q < \infty$) for the estimation of bounded variation functions on $[0,1]^d$ with $d\ge 1$;  Such rates in $L^2$-risk for the bounded variation space match those for the inscribed Sobolev space $W^{1,1}$ when $d \le 2$, but turn out to be slower for $d \ge 3$ (a phase transition). Similar statistical justifications of {the estimator in} \eqref{eq:combined} have been established for linear inverse problems in \citet{dAMu20}. 

\subsection{The (group) lasso and Nemirovskii's estimator}
\label{sec:lasso}

Let $\kl{\phi_\la}_{\la \in \La}$ be an orthonormal basis of $ \R^p$. Suppose that
the index set $\La$ is written as
\begin{equation*}
	\La  =  \bigcup_{\kk \in \KK}  \La_\kk \,,
\end{equation*}
with possible overlapping (though the later theory assumes disjointness) subsets $\La_\kk \subset \La$. We further define, for every $\kk \in \KK$, the linear mapping $\base_\kk\colon \R^p \to \R^{\La_\kk}$,
\begin{equation*}
\base_\kk\beta
\coloneqq
\kl{\inner{\phi_\la}{\beta} : \la \in \La_\kk}
\quad \text{ for } \beta \in \R^p \,.
\end{equation*}
The \emph{adaptive group lasso} is then defined via the penalized optimization (cf.\ \eqref{eq:rss-pen})
\begin{equation} \label{eq:group-lasso}
   \minimize_{\beta \in \R^p}  \frac{1}{2} \norm{Y - \X\beta}^2_2
     +
     \gamma \sum_{\kk \in \KK}  w_\kk \snorm{\base_\kk\beta}_2 \,.
\end{equation}
The following theorem (based on Fenchel duality) shows that the adaptive group lasso estimator
is also a special instance of the MIND in \eqref{eq:mind}.

\begin{theorem}[Dual formulation of the adaptive group lasso] \label{thm:group-lasso}
Suppose that $\La  =  \bigcup_{\kk \in \KK}  \La_\kk$ consists of disjoint subsets $\La_\kk$.
Then the adaptive group lasso in \eqref{eq:group-lasso} and the
constrained optimization problem
\begin{equation} \label{eq:mre-gl}
	\left\{
	\begin{aligned}
	&\minimize_{\beta \in \R^p}  &&
	\frac{1}{2}
	\norm{ \X \beta }_2^2   \\
	&\st &&
	\max_{\kk\in\KK}
	\frac{\norm{ \base_\kk\X^{\mathsf{T}} \skl{ Y - \X \beta } }_2}{w_\kk}
	\leq   \gamma
	\;,
	\end{aligned}
	\right.
\end{equation}
have the same sets of solutions.
\end{theorem}

In the special case that each subset $\La_\kk$ consists of a single element, the adaptive group lasso estimator equals the adaptive lasso estimator \citep{Zou06,HuaMaZha08}. If additionally the weights are taken equal to one, then one obtains the standard lasso estimator
\begin{equation*} 
   \minimize_{\beta \in \R^p}  \frac{1}{2}\norm{ Y - \X \beta   }_2^2
    +
    \gamma
    \sum_{j=1}^p
    \abs{ \beta_j } \,,
\end{equation*}
which is introduced in \cite{Tib96}. 

The following special case of \cref{thm:group-lasso} is well known, also in the compressed sensing and sparse recovery community \citep{Fuc01,Fuc04,Tro06}. A related result has been obtained in \cite{OsbTur00} for the lasso with the $\ell^1$-norm as constraint.

\begin{corollary}[Dual formulation of the lasso] \label{cor:lasso-dual}
The standard lasso estimator and the
constrained optimization problem
\begin{equation*} 
	\left\{
	\begin{aligned}
	&\minimize_{\beta \in \R^p}  &&
	\frac{1}{2}  \norm{ \X \beta }_2^2   \\
	&\st &&
	\norm{\X^{\mathsf{T}} \skl{ Y - \X \beta } }_\infty
	\leq   \gamma
	\;,
	\end{aligned}
	\right.
\end{equation*}
have the same sets of solutions.
\end{corollary}

The dual characterization of the lasso in \cref{cor:lasso-dual} reveals the close connection between the Dantzig selector and the lasso: Both  estimators use the same  constraint $\norm{\X^{\mathsf{T}} \skl{ Y - \X \beta } }_\infty \leq \gamma$. This connection has been exploited in \citet{BiRT09}.  However, among all feasible elements the Dantzig selector minimizes the $\ell^1$-norm of  $\beta$, whereas the lasso minimizes the $\ell^2$-norm of the prediction $\X \beta$. Hence, in general, the lasso is better for prediction, while the Dantzig selector is better for coefficient estimation.

As a further consequence of \cref{thm:group-lasso}, we can show that in the case of $\X=\Id$ the adaptive group lasso estimator is equal to the block soft-thresholding.

\begin{corollary}[Equivalence of the block thresholding and the adaptive group lasso]\label{thm:bt-gl}
Consider the regression case $\X = \Id$ and suppose that $\La  =  \bigcup_{\kk \in \KK}  \La_\kk$ consists of disjoint subsets $\La_\kk$. Then the following three estimators coincide:
\begin{enumerate}[label=\emph{(\alph*)}] 
\item
The block soft-thresholding estimator in \eqref{eq:block-st2}.
\item
The  adaptive group lasso estimator in \eqref{eq:group-lasso}.
\item
The  multiscale Dantzig estimator in \eqref{eq:mre-gl}.
\end{enumerate}
\end{corollary}

Finally, we stress that in exactly the same way Nemirovskii's estimator  can be shown to be equivalent to its reverse constrained variant.

\begin{theorem}[Reverse formulation of Nemirovskii's estimator]\label{th:rvnem}
Nemirovskii's estimator in \eqref{eq:nemirovski} is equivalent to the reverse formulation 
$$
\left\{
	\begin{aligned}
	& \minimize_{\beta \in \R^n} \norm{\beta}_{k,q} \\
	&\st \norm{Y-  \beta}_{{\cal N}} \leq q,
	\end{aligned}
\right.
$$
for some proper choice of $q$.
\end{theorem}

Thus, Nemirovskii's estimator is also an instance of MIND in \eqref{eq:mind}. The reverse version was indeed introduced by \citet{Nem85}, who credited the original idea to S.~V.~Shil'man. It was shown to be adaptively minimax optimal over Sobolev balls in \citet{GrLM18}. The reverse Nemirovskii's estimator appears to us favorable over the original and penalized version, as the threshold $q$ has a distinct statistical interpretation~(\cref{sec:MIND}).

\subsection{Multiscale change point segmentation}\label{ss:mcps}

As a special case, we consider now the model \ref{M3} of change point detection in detail. The target is to estimate a piecewise constant function $f : [0,1]\to \R$ with values $f(x_i) = \sum_{k=1}^i \beta_k$ for $i = 1,\ldots, n$. A typical choice of regularity measure is the number of jumps, that is, $R(\beta) = \norm{\beta}_0$, which is defined as the number of non-zero elements of $\beta$. The corresponding  \emph{jump penalized least squares} estimator \citep{BoyKemLieMunWit09} takes the form of 
$$
\minimize_{\beta \in \R^p}  \frac{1}{2}\norm{ Y - \X \beta}_2^2
    +
    \gamma
 \norm{\beta}_0\,.
$$ 
{Refined} penalties can be found in e.g.\ \citet{ZhSi07} and \citet{DaYa13}. 

In practice, the selection of the global penalty parameter $\gamma$ is tricky for jump penalized least squares, in particular, when the change points are spatially inhomogeneous. \citet{FriMunSie14} introduced a remedy, SMUCE (simultaneous multiscale change point estimator) following the MIND idea, by combining variational estimation with multiple tests on residuals over different scales. More generally, \emph{multiscale change point segmentation} (MCPS; \citealp{LiGM19}) is defined as any solution to the constrained non-convex optimization problem
\begin{equation}\label{eq:mcps}
\left\{
\begin{aligned}
	&\minimize_{\beta \in \R^n}  \quad \norm{\beta}_0   \\
	&\st \quad
	\max_{(i,j) \in \mathcal{I}_\beta}\left \{
	\frac{\abs{\sum_{k = i}^j(Y - \X \beta)_k}}{\sqrt{j-i+1}} - s_{i,j}  \right \}
	\leq   q\,,
	\end{aligned}
\right.
\end{equation}
where $\mathcal{I}_\beta = \set{(i,j)\, :\, 1 \le i \le j \le n, \text{ and } \beta_k = 0 \text{ for }i < k \le j}$, and $s_{i,j}$ are certain scale penalties, for instance, $s_{i,j} = \sqrt{2\log (n/(j-i+1))} $. In particular, SMUCE and its FDR variant, FDRSeg \citep{LiMS16}, are instances of the MCPS. It can be easily seen that every MCPS is in fact a special instance of the MIND in \eqref{eq:mind}. Similar to the general strategy for MIND, the threshold $q$ can be set as the $(1-\alpha)$-quantile of 
\begin{equation}\label{eq:msintv}
T_n = \max_{1 \le i \le j \le n} \left\{\frac{\abs{\sum_{k = i}^j \eps_k}}{\sqrt{j-i+1}} - s_{i,j}\right\}\,,
\end{equation}
which, under no assumption, guarantees uniformly over $\beta$ and $f = \X \beta$ that
$$
\Po_\beta\set{\snorm{\hat\beta}_0 \le \norm{\beta}} = \Po_f \set{\# \text{jumps of }\hat f_n \;\le\; \# \text{jumps of }f_n} \ge 1-\alpha\,,
$$
where $\hat\beta$ is computed by the MCPS, {$\hat f_n = \X \hat\beta$ and $f_n = \X \beta$ (recall \eqref{eq:ssg} in \cref{sec:MIND}).}  
 {A comprehensive discussion of statistical optimality} properties are provided in \citet{FriMunSie14,LiMS16} and \citet{LiGM19}. Extensions can be found e.g.\ in \citet{PeSM17} for heterogeneous Gaussian error, in \citet{VaBM21} for general independent data, in \citet{DeEV20} for dependent data, and in \citet{LMSW20} for {automatic} selecting the bins in a histogram and exploratory data analysis. Besides, \citet{BeHM18} extended the MCPS to {blind source separation, more precisely, to} recover piecewise constant functions (taking values in a finite set) from noisy measurements of their mixtures. In addition, estimators similar to MCPS, but with different regularizations accounting for shape or smoothness have been considered in \citet{DavKov01,DueSpo01,DavMei08,DavKovMei09} and \citet{SchMunDue13}.

\section{Distributional properties and selection of the threshold}
\label{sec:dist}

As discussed in \cref{sec:MIND}, the threshold $q$ for MIND is fully determined by the quantiles of multiscale statistic $T_n$ in \eqref{eq:mstat}, which can be estimated by Monte Carlo simulations. As an alternative and computationally more efficient approach, such quantiles can be computed via either limiting distributions of $T_n$ or (approximate) tail probabilities and bounds of~$T_n$.  Here we focus on the first approach, while for the latter we refer to e.g.\ \citet{SieYak00}, \citet{FaLS20} and the references therein.

\subsection{Case of orthogonal bases}

Suppose for the moment that $\skl{\phi_{n,\la}}_{\la \in \La_n}$ is an orthonormal basis of $\R^n$. A particular important class of probe  functionals takes $\T_{n,\la} \eps_n  \coloneqq \sinner{ \phi_{n,\la}}{ \eps_n }$. Under the i.i.d.\ Gaussian assumption, {these functionals are i.i.d.\ Gaussian again, and hence}
\begin{multline} \label{eq:G1}
\lim_{n \to \infty}
\wk\Bigl\{
\max_{\la \in \La_n}
\sabs{ \sinner{\phi_{n,\la}}{\eps_n}}
\leq
\sigma  \sqrt{2 \log n}
+
\sigma \,\frac{2x-  \log\log n - \log\pi}{2\sqrt{2 \log n}}
\Bigr\}
=
\exp\kl{-e^{-x} } \,,
\end{multline}
where the limit distribution is known as the \emph{Gumbel extreme value distribution}. In the first order, $q$ equals $\sigma \sqrt{2 \log n}$, which corresponds to the asymptotic behavior of  the maximum of  absolute values of i.i.d.\ Gaussian random variables, and  is the universal threshold as
proposed in the seminal work {of}~\cite{Don95}, recall \cref{ss:wst}.

\subsection{Redundant systems}

Now suppose that $\skl{\phi_{n,\la}}_{\la \in \La_n}$ is a redundant frame instead of an orthonormal basis and consider again the probe functionals  $\T_{n,\la} \eps_n  \coloneqq \sinner{ \phi_{n,\la}}{ \eps_n }$.
In this situation finding the distribution of the multiscale statistic $T_n$ in \eqref{eq:mstat} is more involved than in the independent case. In  \cite{HalMun13} similar asymptotic distributions as in \eqref{eq:G1} have been derived for a wide class of redundant systems:

\begin{definition}[Asymptotically stable frames]\label{def:asf}
For any $n \in \N$,  let $ \mathcal D_n \coloneqq \skl{\phi_{n,\la}}_{\la \in \La_n}$
be a frame of $\R^n$ with upper frame bound $b_n$, as in \eqref{eq:frmbnd}. Then
$\set{\mathcal D_n}_{n \in \N}$ is called an \emph{asymptotically stable} family of frames,
if $\mnorm{\phi_{n,\la}}_2 =  1$ for all $n \in \N$ and all  $\la \in \La_n$, $\sup\set{ b_n : n \in \N}< \infty$, and $\abs{\sset{\skl{\la,\mu} :
\sabs{\sinner{\phi_{n,\la}}{\phi_{n,\mu}}}   \geq \rho  }}
= o\bigl({\abs{\La_n}}/{\sqrt{\log \sabs{\La_n}}} \bigr)$ for some $\rho < 1$.
\end{definition}

Roughly speaking, for such frames, correlations of  $\sinner{\phi_{n,\la}}{\eps_n}$
asymptotically vanish fast enough and hence the system $\skl{\phi_{n,\la}}_{\la \in \La_n}$ asymptotically behaves as an orthonormal system. Many frames used in applications, such as unions of bases, non-redundant and redundant wavelet systems, and curvelet frames are covered by \cref{def:asf} (\citealp{HalMun13}). Hence, this justifies the universal thresholding {by $\sigma\sqrt{2\log \abs{\La_n}}$} for many systems beyond wavelets.

In case of one dimension and the probe functionals as indicators of all subintervals, the multiscale statistic $T_n$ takes the particular form of \eqref{eq:msintv}. Here the probe functionals are a strongly redundant  frame, and thus not asymptotically stable. However, similar asymptotic distributional results as in \eqref{eq:G1} still hold, see \citet{SieVen95} and \citet{SieYak00}. Generalization to higher dimension can be found in \citet{Kab11,PrWM18} and \citet{KoMW20}.

\section{Numerical computation}\label{sec:num}

In the previous sections we have seen that many estimation techniques can be written as instances of the MIND in \eqref{eq:mind}. For thresholding methods, the solution of \eqref{eq:mind} is often given by an explicit formula (see \cref{sec:threshold}). In the general case, however, a MIND must be computed numerically by
applying an optimization procedure.

In case that the regularization functional $R$ is convex, \eqref{eq:mind} is a nonsmooth convex optimization, and often of large size (when there are many probe functionals). Recent development (e.g.\ \citealp{Nes05,BecTeb09,BecBibCan11,ChaPoc11}) in optimization has made the computation of MIND in \eqref{eq:mind} feasible (e.g.\ in a few minutes for $256 \times 256$ images) on standard laptops, see e.g.\ \citet{FriMarMun12b}. Run time comparisons suggest the primal-dual hybrid-gradient algorithm \citep{ChaPoc11} as a powerful general computational scheme for MIND, see \citet{ALMW20} for details and a comparison to other optimization methods, including semismooth Newton and ADMM (alternating direction method of multipliers). MATLAB codes are freely available at \url{https://github.com/housenli/MIND}. 

In case of nonconvex $R$, the same algorithm can be applied but there is no guarantee for global optimality, in general. However, in the particular case of change point segmentation (see \cref{ss:mcps}), the global optimal solutions for MCPS in \eqref{eq:mcps} can be computed using dynamic programming algorithms together with speedups leading to run time of the order $O(n \log n)$ {in most often cases.} Implementation is made available in R packages (e.g.\ \texttt{stepR} or \texttt{FDRSeg}) on \hyperlink{https://cran.r-project.org/}{CRAN}, see \citet{FriMunSie14,LiMS16} and \citet{PeSM17} for details. Besides, approximate solutions can be found in sublinear run time \citep{KLHMB20}.

\subsection{Image denoising}
We consider an example of model \ref{M1}, where $\beta \in \R^n$ is an $m \times m$ image with $n = m^2$, often known as \emph{image denoising}. We numerically compare several instances of MIND: wavelet soft-thresholding (\cref{ss:wst}), total variation penalized least squares (\cref{ss:pls}), and two hybrid approaches (\cref{ss:hdtv}) combining total variation with wavelets and shearlets, respectively.  Daubechies’ symlets with six vanishing moments and shearlets with four scale levels (default in \citealp{KutMorZhu12}) are used.  The threshold for wavelet soft-thresholding and hybrid methods is set as the 90\%-quantile of the corresponding multiscale statistic in \eqref{eq:mstat}, respectively. The penalty parameter $\gamma$ in  total variation penalized least squares is tuned to give the best visual quality {(cf.\ \cref{sss:pe}).} 

Results are depicted in \figref{fig:berkeley}. In this example, wavelet soft-thresholding performs the worst; In particular, it shows artificial oscillations across discontinuity (common for wavelets due to missing band pass information). The total variation as a regularization functional has a comparably better performance, while blurring out several details (e.g.\ at the bottom right conner). The hybrid combination with wavelets leads to an improved performance, because more than one spatial scale is incorporated (recall \cref{ss:hdtv}). The hybrid method of total variation and shearlets produces clearly the best result as it recovers features over a range of spatial scales quite accurately. It seems that visual inspection of the image is compatible with the PSNR (peak signal-to-noise ratio).

\begin{figure}
\centering
\includegraphics[width=0.8\textwidth]{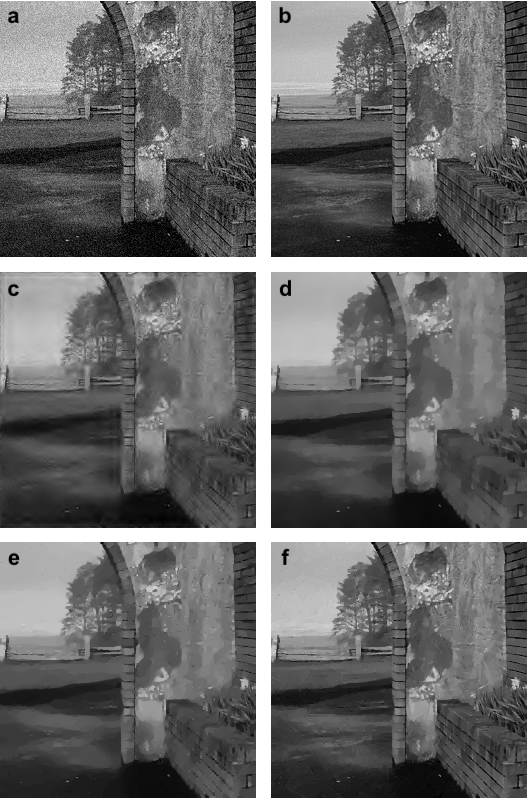}
\caption{Image denoising via different MINDs. (a) Noisy data with $\sigma = 14$ (PSNR = 25.2); (b) True image from BSDS500 \citep{MartinFTM01}; (c) Wavelet soft-thresholding (PSNR = 25.4); (d) Total variation penalized least squares (ROF; PSNR = 25.9); (e) Hybrid wavelet--total variation (PSNR = 26.1); (f) Hybrid shearlet--total variation (PSNR = 28.2).  
}
\label{fig:berkeley}
\end{figure}

\begin{figure}
\centering
\includegraphics[width=0.95\textwidth]{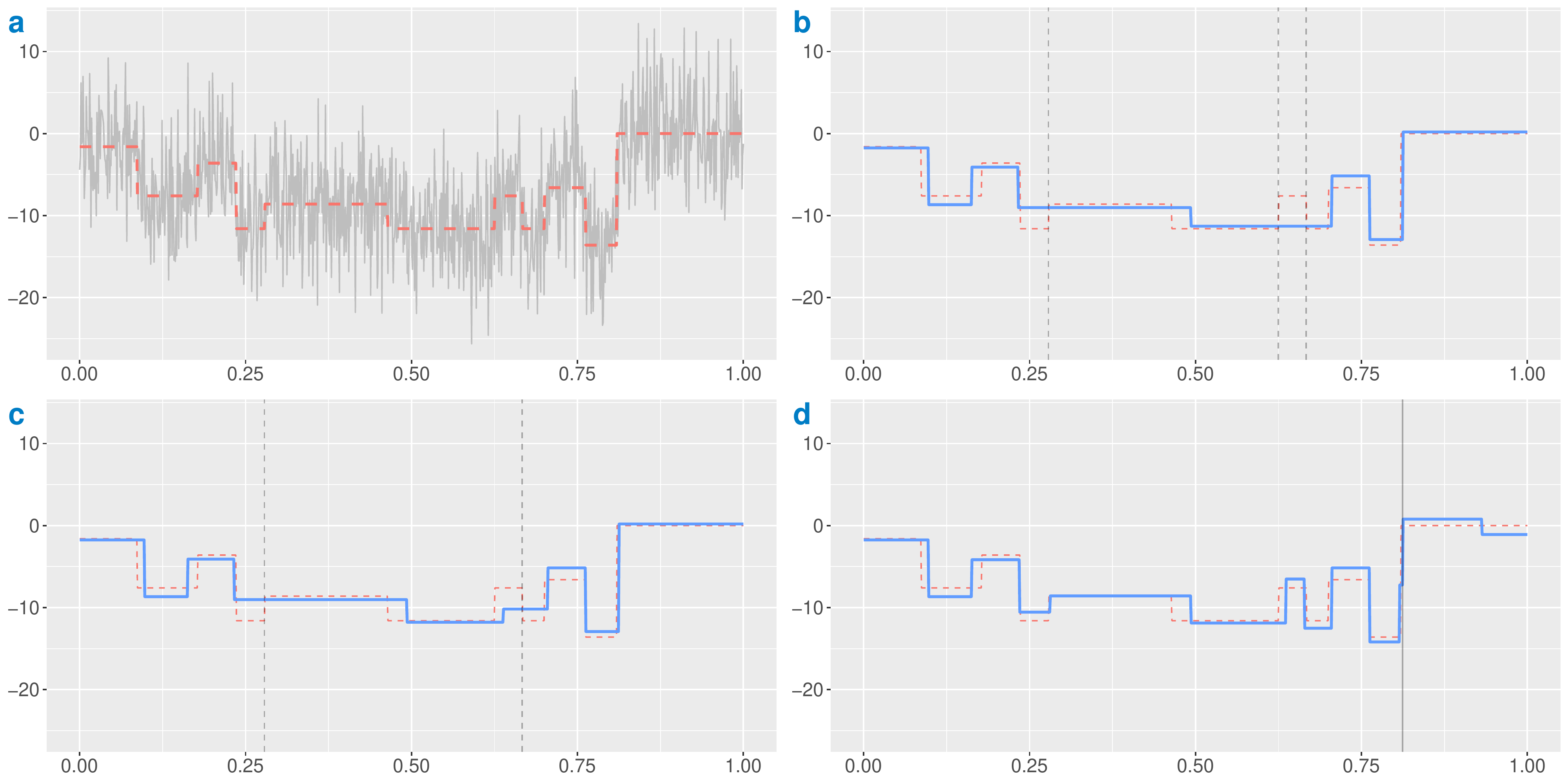}
\caption{Change point estimation for a randomly generated signal. (a) Noisy data (solid gray line); (b)--(d) Estimators (solid blue line) of jump penalized least squares, SMUCE, and FDRSeg, respectively. Vertical dashed lines indicate missed change points (false negatives), while vertical solid lines mark artificial change points (false positives). The true signal (dashed red line) is plotted in all panels.}
\label{fig:changepoint}
\end{figure}

\subsection{Detection of change points}
Recall example \ref{M3}. On a randomly generated piecewise constant signal, we run the jump penalized least squares, and two multiscale methods, SMUCE \citep{FriMunSie14} and FDRSeg \citep{LiMS16}, all of which are instances of the MIND (\cref{ss:mcps}). The penalty parameter $\gamma = 2\sigma^2\log_2 n$, which corresponds to the Bayesian information criterion, is used for the jump penalized least squares. For SMUCE and FDRSeg, the default parameters in R packages \texttt{stepR} and \texttt{FDRSeg}, respectively, are used. The comparison is shown in \figref{fig:changepoint}. The jump penalized least squares recovers major structure of the signal, but misses three jumps. SMUCE examines over multiple scales (here lengths of intervals) and recovers one more jump. By switching to a weaker error criterion, FDRSeg detects all jumps of the signal, but at the price of including one artificial jump.

\section*{Acknowledgement}
We thank Shayan Hundrieser, Russell Luke, Marc Teboulle, Frank Werner for helpful comments and {in particular} Timo Aspelmeier for inspiring discussions and computational assistance.

\appendix

\section{Relations between constrained and unconstrained minimization}\label{s:rlvf}

There are at least two close connections between the constrained optimization in \eqref{eq:rss-const} and the unconstrained optimization in \eqref{eq:rss-pen}. The first one arises from a Lagrangian multiplier approach, whereas the second arises from Fenchel duality. In what follows we analyze both approaches.

\subsection{Lagrangian multiplier approach}
\label{sec:lagrange}

Let $\rf \colon \R^p \to \R \cup \set{\infty}$ and $G \colon \R^n \to \R \cup \set{\infty}$ be convex proper functionals, and let $\X \in \R^{n \times p}$.   It is convenient to allow a convex functional to attain the value $\infty$. The set $\mathcal D (\rf) \coloneqq \set{\beta \in \R^p \colon \rf(\beta) < \infty}$ where the functional
$\rf$ takes  finite values is referred to as the \emph{domain} of $\rf$.  The functional $\rf$ is called \emph{proper}, if
$\mathcal D (\rf) \neq \emptyset$, see \citet{Roc70}.

We will investigate the relation between the following optimization problems:
\begin{align}  \label{eq:a1}
& \min \set{\rf\skl{\beta}  \colon
G\kl{\X \beta  } \leq   q} && \text{($G$-constrained)} \\ \label{eq:a2}
&\min \set{ G \kl{\X \beta } + \gamma
\rf\skl{\beta}  \colon \beta \in \R^{p}} && \text{(penalized)} \\\label{eq:a3}
&\min \set{G \kl{ \X\beta }  \colon
\rf\skl{\beta} \leq   c}  &&\text{($R$-constrained)},
\end{align}
for parameters  $q$, $\gamma$, $c  \geq 0$. Here we simplify the notation by $G\kl{\X \beta  } \equiv G\kl{\X \beta; Y }$.  We will show that  \eqref{eq:a1}, \eqref{eq:a2}, and \eqref{eq:a3} are  equivalent in the sense that any solution $\beta_\star$ of one of these optimization problems also provides  a solution of the others.

The proof will use the following saddle-point theorem, which is a well known result from convex analysis.

\begin{lemma}
Suppose that\label{lem:kt} $\Phi, \Psi \colon \R^p \to \R \cup \set{\infty}$ are proper convex functionals and denote by  $\mathcal L \colon \R^p \times \R_{\geq 0} \to \R \cup \set{\infty}$ the Lagrangian defined by  $ \mathcal L \skl{\beta, \mu} = \Phi\skl{\beta} + \mu \Psi\skl{\beta }$.

\begin{enumerate}[label=\emph{(\alph*)}] 
\item\label{lem:kt:a}
Let $\skl{\beta_\star, \mu_\star} \in \R^p \times \R_{\geq 0}$  be a \emph{saddle-point}
of $ \mathcal L$, that is,   for all $ \kl{\beta, \mu} \in \R^p \times \R_{\geq 0}$,
 \begin{equation*}
 \mathcal L \skl{\beta_\star, \mu}
 \leq
 \mathcal L \skl{\beta_\star, \mu_\star}
 \leq
 \mathcal L \skl{\beta, \mu_\star} \,.
 \end{equation*}
 Then $\beta_\star \in \argmin \set{\Phi\kl{\beta} \colon \Psi\kl{\beta } \leq 0 }$.

\item\label{lem:kt:b}
If, additionally, Slater's condition
\begin{equation*}
\mathcal D \kl{\Phi}  \cap \set{\beta \in \R^p: \Psi (\beta) < 0}
\neq \emptyset
\end{equation*}
holds and $\beta_\star \in \argmin \{\Phi\kl{\beta} \colon \Psi\kl{\beta } \leq 0 \}$, then
there exists a Lagrangian multiplier $\mu_\star\geq 0$ such that
$\skl{\beta_\star, \mu_\star}$  is a saddle-point of $ \mathcal L$.
\end{enumerate}
\end{lemma}

\begin{proof}
Part \emph{\ref{lem:kt:a}} follows from Theorem~28.3 in \citet{Roc70}. 

Part \emph{\ref{lem:kt:b}} follows from Corollary~28.2.1 and Theorem~28.3 in \citet{Roc70}.
\end{proof}

We next derive the equivalence of the $G$-constrained and penalized formulation, i.e.\ of \eqref{eq:a1} and   \eqref{eq:a2}.

\begin{theorem}[Equivalence of \eqref{eq:a1} and \eqref{eq:a2}]\label{thm:lagrange}
Let $\rf \colon \R^p \to \R\cup \set{\infty}$ and $G\colon \R^n \to \R\cup \set{\infty}$
be proper convex functionals and let $\X\in \R^{n\times p}$.
Then \eqref{eq:a1} and \eqref{eq:a2} are equivalent in the following sense:
\begin{enumerate}[label=\emph{(\alph*)}] 
\item
\label{th:lag:a}
If $\beta_\star$ is a solution of
\eqref{eq:a2} for some $\gamma >0$,
then $\beta_\star$ solves
\eqref{eq:a1} with
$q = G\kl{\X \beta_\star}$.

\item
\label{th:lag:b}
For a given $q>0$, suppose that $G(\X \beta_0) < q$ for  some
$\beta_0 \in \mathcal D (\rf)$ and  that we have
\begin{equation} \label{eq:ass-dist}
\argmin (\rf)  \cap \set{ \beta \colon  G \kl{\X \beta } \leq  q}  = \emptyset \,.
\end{equation}
Then, any  solution  $\beta_\star$ of  \eqref{eq:a1}
satisfies $G\kl{ \X \beta_\star } = q$ and there exists some
$\gamma >0$ such that $\beta_\star$ also solves~\eqref{eq:a2}.
\end{enumerate}
\end{theorem}

In the special case of strictly convex penalties $\rf$ and $ G \kl{\X \beta } = \frac{1}{2}\norm{Y - \X \beta}^2$, \cref{thm:lagrange} seems to be well known but it is hard to find an accessible
reference. For an infinite dimensional setting such a  proof can  be found in \citet{Vas70} and 
\cite{IvaVasTan02}. Below we include a short proof for the convex penalties (not necessarily strict convex), but for finite dimensional ground space.  Note that if $\rf$ is strictly convex and $z_0$ denotes  its unique minimizer, then the
 condition of \eqref{eq:ass-dist}  reads $G\kl{ \X z_0}  > q$, which simply means that the unique minimizer  of $\rf$ is not feasible for \eqref{eq:a1}.

\begin{proof}[Proof of \cref{thm:lagrange}]
Part \emph{\ref{th:lag:a}}:
Suppose first that  $\beta_\star$ is a solution of
\eqref{eq:a2} for some $\gamma > 0$ and take $q = G\kl{\X \beta_\star}$,
$\Psi(\beta)= G\kl{\X \beta} -q$. Then we have $\Psi\kl{ \beta_\star} =0$ which implies
that
$\rf \mkl{ \beta_\star } + \mu \Psi\mkl{\beta_\star} =
\rf \mkl{\beta_\star } + \frac{1}{\gamma}\Psi\mkl{\beta_\star}
$ for all $\mu \geq 0$.
 Further, since $\beta_\star$ is a minimizer of $G(\X \beta; Y) + \gamma \rf(\beta)$, we also
 have $ \gamma \rf(\beta_\star) + \kl{G\kl{\X \beta_\star} -q}
\leq  \gamma \rf(\beta) + \kl{G\kl{\X \beta} -q}$, which yields
$ \rf(\beta_\star) + \frac{1}{\gamma } \Psi\kl{\beta_\star}
\leq  \rf(\beta) + \frac{1}{\gamma } \Psi\kl{\beta}$ for all $\beta \in \R^p$.
In total, we have verified that
\begin{equation*}
	\rf \mkl{ \beta_\star }
	+
    \mu \Psi\mkl{\beta_\star}
	\leq
    \rf \mkl{\beta_\star }
	+
    \frac{1}{\gamma}
    \Psi\mkl{\beta_\star}
    \leq
	\rf \kl{\beta }
	+ \frac{1}{\gamma} \Psi\kl{\beta}
\quad \text{ all }
\kl{\beta, \mu} \in \R^p \times \R_{\geq 0}\,.
\end{equation*}
Hence $\skl{\beta_\star, 1/\gamma}$ is a saddle-point of
the Lagrangian $\mathcal L\kl{\beta,\mu} = \rf \kl{\beta } +
\mu \Psi\kl{\beta}$.  \cref{lem:kt} implies that $\beta_\star
\in \argmin \set{ \rf\skl{\beta} \colon \Psi\kl{\beta} \leq 0}$. Recalling
that $\Psi(\beta)= G\kl{\X \beta} -q$ shows that $\beta_\star$ is a
solution of \eqref{eq:a1}.

Part \emph{\ref{th:lag:b}}: Now suppose that $\beta_\star \in \R^p$ is a solution
of  the constrained problem in \eqref{eq:a1} for some $q>0$.
According to our assumption there exists $\beta_0 \in \R^p$
with $G \kl{\X \beta_0} < q$. Consequently,  the convex functionals
$\rf\skl{\beta}$ and  $\Psi(\beta)= G \kl{\X \beta} - q$ satisfy Slater's condition.
 According  to \cref{lem:kt} the corresponding Lagrangian $\mathcal L \kl{\beta, \mu}
 = \rf(\beta) +  \mu \Psi\kl{\beta}$ admits
 a saddle-point $\skl{\beta_\star, \mu_\star} \in \R^p \times \R_{\geq 0}$.
In particular, $\beta_\star$ minimises $\rf\skl{\beta}
+ \mu_\star  G \kl{\X \beta}$. It remains to show that $\mu_\star >0$. To that end, suppose  to the contrary that
$\mu_\star =0$. Then $\beta_\star \in \argmin \rf$.
Since  $\beta_\star$ is also a solution of  \eqref{eq:a1}, it in particular satisfies the constraint in
\eqref{eq:a1}, which implies
\begin{equation*}
\beta_\star \in \argmin (\rf)  \cap \set{ \beta \in \R^p \colon  G\kl{\X \beta } \leq  q}
\end{equation*}
This contradicts the assumption in \eqref{eq:ass-dist}.  Hence $\mu_\star >0$   and further
$G\kl{ \X \beta_\star }= q$.
This concludes the proof of the theorem
after taking $\gamma = 1/\mu_\star$.
\end{proof}

Next  we show the equivalence between the penalized and $R$-constrained formulation, i.e.\  \eqref{eq:a2} and \eqref{eq:a3}.
Together with  \cref{thm:lagrange}
this also implies that the $G$- and $R$-constrained formulation are equivalent, too.

\begin{theorem}[Equivalence of \eqref{eq:a2} and \eqref{eq:a3}]\label{thm:lagrange-b}
Let $\rf\colon \R^p \to \R\cup \set{\infty}$ and $G\colon \R^n \to \R\cup \set{\infty}$
be proper convex functionals and let $\X\in \R^{n\times p}$.
Then \eqref{eq:a2} and \eqref{eq:a3} are equivalent in the following sense:
\begin{enumerate}[label=\emph{(\alph*)}] 
\item
\label{th:lag-b:a}
If $\beta_\star$ is a solution of
\eqref{eq:a2} for some $\gamma >0$,
then $\beta_\star$ solves
\eqref{eq:a3} with $c = \rf(\beta)$.

\item
\label{th:lag-b:b}
For a given $c>0$, suppose $\rf(\beta_0) < c$ and $G\kl{\X\beta_0}< \infty$
for  some $\beta_0 \in \R^p$ and  let
$\beta_\star$ be a solution of \eqref{eq:a3}. Then
there exists $\gamma \geq 0$  such that  $\beta_\star$
solves~\eqref{eq:a2} and $\gamma  ( \rf(\beta_\star) - c ) = 0$.
\end{enumerate}
\end{theorem}

\begin{proof}
Part \emph{\ref{th:lag-b:a}}: 
Suppose first that  $\beta_\star$ is a solution of
\eqref{eq:a2} for some $\gamma > 0$ and take $c = \rf\kl{\X \beta_\star}$,
$\Phi\kl{\beta} = G\kl{\X\beta}$, $\Psi(\beta)= \rf\skl{\beta} - c$.
Then we have $\Psi \kl{\beta_\star} =0$ which implies that
$\Phi \kl{ \beta_\star } + \mu \Psi\kl{\beta_\star} =
\Phi \mkl{\beta_\star } + \gamma \Psi\kl{\beta_\star}
$ for all $\mu \geq 0$.
 Further, since $\beta_\star$ is a minimizer of $G(\X \beta) + \gamma \rf(\beta)$, we also
 have $ G(\X \beta_\star) + \gamma \kl{\rf(\beta_\star)-c}
\leq  G(\X \beta) + \gamma \kl{\rf(\beta)-c}$
for $\beta \in \R^p$, which yields
$ \Phi(\beta_\star) + \gamma \Psi\kl{\beta_\star}
\leq  \Phi(\beta) + \gamma \Psi\kl{\beta}$.
In total, we have verified the inequalities
\begin{equation*}
	\Phi(\beta_\star) + \mu \Psi\kl{\beta_\star}
	\leq
 \Phi(\beta_\star) +
 \gamma \Psi\kl{\beta_\star}
    \leq
 \Phi(\beta) +
 \gamma \Psi\kl{\beta}
\quad\text{for }\kl{\beta, \mu} \in \R^p \times \R_{\geq 0}  \,.
\end{equation*}
Hence  $\skl{\beta_\star,  \gamma}$ is a saddle-point of
the corresponding Lagrangian $\mathcal L\kl{\beta,\mu} = \Phi \kl{\beta } +
\mu \Psi\kl{\beta}$. \cref{lem:kt} therefore implies that
$\beta_\star \in \argmin \sset{ \Phi \kl{\beta } \colon  \Psi\kl{\beta} \leq 0 }$. Recalling that $\Phi\kl{\beta} = G\kl{\X\beta}$, $\Psi(\beta)= \rf\skl{\beta} -c$ shows that
$\beta_\star$ is a solution of \eqref{eq:a1}.

Part \emph{\ref{th:lag-b:b}}: 
Suppose that $\beta_\star \in \R^p$ is a solution
of \eqref{eq:a3} of some $c>0$.
According to our assumption there exists $\beta_0 \in \R^p$
with $\rf \kl{\beta_0} < q$. Consequently,  the convex functionals
$\Phi \kl{\beta}  = G(\X \beta)$ and $\Psi(\beta)= \rf\skl{\beta} -c$ satisfy  Slater's
condition.
According  to \cref{lem:kt} the corresponding Lagrangian
 $\mathcal L \kl{\beta, \mu} = G(\X \beta) +  \mu \kl{\rf\skl{\beta} -c}$ admits
 a saddle-point $\skl{\beta_\star, \gamma} \in \R^p \times \R_{\geq 0}$.
In particular, $\beta_\star$ minimises $G(\X \beta) +  \gamma \rf\skl{\beta}$.
 Since $\beta_\star$ is feasible, we have  $\rf \kl{\beta_\star} - c \leq 0$.
 Further, for all $\mu \geq 0$,
\begin{equation*}
	G(\X \beta_\star) +
	\mu \kl{\rf \kl{\beta_\star} -c}
	\leq
	 G(\X \beta_\star) +
	 \gamma \kl{\rf \kl{\beta_\star} -c},
 \end{equation*}
 which implies $\gamma(\rf \kl{\beta_\star} -c) =0$.
\end{proof}

We emphasize again that the relation between the threshold
$q$ and the penalty $\gamma$ in \cref{thm:lagrange} is
only given in an implicit manner. Therefore, the unconstrained problem  in \eqref{eq:a2} cannot be immediately used for solving the
(more difficult) constrained problem in \eqref{eq:a1}. The same statement applies to the
equivalence of the penalized and $R$-constrained formulation. 
As shown in the next subsection, a more explicit relation between the $G$-constrained formulation and yet another
different  unconstrained optimization problem can be derived from the
Fenchel duality.

The derived equivalences in particular apply to the MIND in \eqref{eq:mind}
by  taking
\begin{equation*}
G(v)
=
\max_{\kk \in \KK}
\set{
\frac{\norm{ \T_{\kk}   \kl{ v - Y }  }_2}{w_\kk} - s_\kk}
\end{equation*}
for $v \in \R^n$, where  $\T_{\kk}$ are given probe functionals,
 $w_\kk$ and $s_\kk$ certain weights and $Y$ the given data,
 see \cref{def:msd}.  Note that the functional $G$  is obviously convex and proper.

\subsection{Fenchel duality}
\label{app:group-lasso}

For a proper convex functional
$\rf \colon \R^p \to \R \cup \set{\infty}$,
we denote by
\begin{align*}
\rf^* \colon \R^p &\to \R \cup \set{\infty} 
\\
\mu
 &\mapsto
 \sup \set{ \inner{\mu}{ \beta} - \rf\skl{\beta}
 : \beta \in \R^p }
\end{align*}
the \emph{Fenchel conjugate}
of $\rf$, {and by $\partial R(\beta_0)$ the \emph{subdifferential} of $R$ at $\beta_0$, i.e.\ $\beta^* \in \partial R(\beta_0)$ if and only if 
$$
R(\beta_0) < \infty\quad \text{ and }\quad \inner{\beta-\beta_0}{\beta^*} + R(\beta_0) \le R(\beta)\, \text{ for all }\beta \in \R^p,
$$
see \cite{Roc70}.}

\begin{definition}
Let  $\rf\colon \R^p \to \R \cup \set{\infty}$
and $G \colon \R^n \to \R \cup \set{\infty}$
be two proper convex functionals and let
$\X \in \R^{n \times p}$.
Then,
\begin{align}
\minimize_{\beta \in \R^p}\; & G \kl{\X \beta} + \rf\skl{\beta}  && \text{\emph{(primal)}}  \label{eq:primal}\\
\min_{\mu \in \R^n}\; & G^*\kl{\mu}  +  \rf^* \skl{-\X^{\mathsf{T}} \mu}    && \text{\emph{(dual)}}    \label{eq:dual}
\end{align}
are referred   to as the primal and  dual minimization problem, respectively,
corresponding to $\rf$, $G$ and $\X$.
\end{definition}

\begin{theorem}[Fenchel's duality theorem]\label{lem:kkt}
Suppose that $\rf\colon \R^p \to \R \cup \sset{\infty}$
and $G \colon \R^n \to \R \cup \sset{\infty}$
are proper convex functionals, and let $\X
\in \R^{n \times p}$.
Suppose further that there exists $\beta_0 \in \R^p$ such that
$\rf\kl{\beta_0} < \infty$ and  that $G$ is continuous at
$\X \beta_0$.
Then, for any $\skl{\beta_\star, \mu_\star} \in \R^p \times \R^n$,
the following statements are equivalent:
\begin{enumerate}[label=\emph{(\alph*)}] 
\item
$\beta_\star$ is a solution of the primal problem in \eqref{eq:primal}
and $\mu_\star$ a solution of the dual problem  in \eqref{eq:dual}.

\item
The Kuhn-Tucker conditions are satisfied, i.e.
\begin{align}\label{eq:kkt1}
    \mu_\star & \in \partial G \skl{\X \beta_\star}
     \\
     \label{eq:kkt2}
  - \X^{\mathsf{T}} \mu_\star &\in \partial \rf\skl{\beta_\star}\,.
\end{align}
\end{enumerate}
\end{theorem}
\begin{proof}
See~\citet[Section~31]{Roc70}.
\end{proof}
If one takes $G\kl{v} = \frac{1}{2} \norm{ Y - v}_2^2$,
then \cref{lem:kkt} yields the following result:

\begin{theorem}
Let\label{lem:2}  $\rf \colon \R^p \to \R \cup \set{\infty}$  be a proper convex functional and let $Y \in \R^n$. Then,
\begin{equation*}
\argmin_{\beta \in \R^p}\set{ \frac{1}{2} \norm{Y-\X \beta }_2^2 +  \rf\skl{\beta}}
= \argmin_{\beta \in \R^p}\set{ \frac{1}{2} \norm{\X \beta}_2^2  +  \rf^*\skl{
\X^{\mathsf{T}}\skl{Y-\X  \beta }}} \,.
\end{equation*}
\end{theorem}

\begin{proof}
If  $\beta_\star$ is a minimizer  of $\frac{1}{2} \snorm{Y-\X \beta }_2^2 +  \rf\skl{\beta}$, then it  is a solution of the primal problem in \eqref{eq:primal}
corresponding to $G\kl{v} = \frac{1}{2} \snorm{ Y - v}_2^2$.
In this case, one easily computes 
$ G^*\kl{\mu} = \frac{1}{2} \snorm{\mu}_2^2  + \sinner{Y}{ \mu}
$, see \citet{Roc70}. Now let  $\mu_\star \in \R^n$ denote a solution of the dual problem in \eqref{eq:dual}.
Since $\rf$ is proper there exists  $\beta_0 \in \R^p$
such that $\rf\kl{\beta_0} < \infty$ and obviously $G$ is continuous
at $\X\beta_0$. Hence we can apply \cref{lem:kkt} which
implies that  the Kuhn-Tucker conditions  in \eqref{eq:kkt1} and \eqref{eq:kkt2} are satisfied.
In particular, \eqref{eq:kkt1} implies  $\mu_\star =  \X \beta_\star - Y $, and therefore
\begin{align*}
	G^* \skl{\mu_\star}
	&=
	G^*\skl{\X \beta_\star - Y}
	\\
	&=
	\frac{1}{2} \snorm{\X \beta_\star - Y}_2^2
	+ \sinner{Y}{ \X \beta_\star - Y}
	\\
	&= \frac{1}{2} \snorm{\X \beta_\star}_2^2 -\frac{1}{2} \norm{Y}_2^2 \,.
\end{align*}
Because  $\mu_\star$ is a solution of the dual problem in \eqref{eq:dual}, this shows that $\beta_\star$ minimizes  $\frac{1}{2} \snorm{\X \beta}_2^2  +  \rf^*\skl{
\X^{\mathsf{T}}\skl{Y-\X  \beta }}$.  Similar arguments  show that the converse relation also holds.
\end{proof}

\section{Proofs}
Here we provide proofs for the assertions in the paper for the sake of completeness, noting that many of them are known but scattered in the literature. These assertions can be seen as special cases of the results obtained in \cref{s:rlvf}. But we favor simple and direct proofs whenever it is possible.

\subsection{Proofs in \cref{sec:threshold}}
In this subsection, we give proofs for the assertions in \cref{sec:threshold} in the paper. Note that the following proofs do not rely on any results from \cref{s:rlvf}. 

\begin{proof}[Proof of \cref{thm:soft}]
Because $\kl{ \phi_\la}_{\la \in \La}$ is an orthonormal basis, we can uniquely write any element in $\R^n$  as linear combination $ \beta  =  \sum x_{\la} \phi_\la$ with coefficients $x_{\la} = \inner{\phi_\la}{\beta}$.   Hence $ \hat \beta  =  \sum \hat x_{\la} \phi_\la$ is a solution of the stated optimization problem if and only if every coefficient $\hat x_{\la}$ is a solution of the one-dimensional optimization problem
 \begin{equation*} 
	\left\{
	\begin{aligned}
	&\min_{x\in \R}  &&   \abs{x_\la}^r   \\
	&\st &&
	\abs{\inner{\phi_\la}{Y} - x_\la }
	\leq   q_\la
	\;.
	\end{aligned}
	\right.
\end{equation*}
The unique minimizer is given by the soft-thresholding
$\hat x _\la  = \eta^{\mathrm{(soft)}} \kl{ \inner{\phi_\la}{Y}, q_\la }$,
which yields the desired characterization using the objective
$\sum_{\la\in \La}\abs{\inner{\phi_\la}{\beta}}^r$ for any $r>0$. 

In the special
case $r=2$, we  have  $\sum_{\la\in \La} \abs{\inner{\phi_\la}{\beta}}^2
= \sum_{i=1}^n \beta_i^2$, which shows the second claim.
\end{proof}

\begin{proof}[Proof of \cref{thm:nng}]
This immediately follows from the variational characterization of soft-thresholding given in \cref{thm:soft} and the relation
\begin{equation*}
	\eta^{\mathrm{(JS)}} \kl{x,q}
	=
	\eta^{\mathrm{(soft)}} \kl{x, \frac{q^2}{\sabs{x}} }
	\, \text{ for }  x \neq 0  \,.
\end{equation*}
\end{proof}

\begin{proof}[Proof of \cref{thm:soft-block}]
The proof is elementary and similar  to the one of  \cref{thm:soft}.
\end{proof}

\begin{proof}[Proof of \cref{thm:wvd}]
By expanding any function $f \in L^2\kl{\Om}$ in a wavelet series $f = \sum_{\la \in \La} x_\la \phi_\la$, with uniquely determined  coefficients $x_\la \in \R$, the stated minimization is equivalent to
\begin{equation*}
	\left\{
	\begin{aligned}
	&\min_x  &&   \sum_{\la\in \La}
	\abs{x_\la}^r   \\
	&\st &&
	\max_{\la \in \La }
	\frac{ \sabs{ \kappa_\la^{-1} \inner{v_\la}{  g  } -  x_\la} }{q_\la}
	\leq   1
	\,.
	\end{aligned}
	\right.
\end{equation*}
The solution of the latter optimization problem is given by component-wise soft-thresholding of $\kappa_\la^{-1}\inner{v_\la}{  g  }$ with thresholds $q_\la$ (see the proof of \cref{thm:soft}). This results in $\hat f^{\rm (WV)}$.
\end{proof}

\subsection{Proofs in \cref{sec:variational}}

In this subsection, we provide proofs for the assertions in \cref{sec:variational} in the paper. They rely on the results from \cref{s:rlvf}. 

To establish \cref{thm:group-lasso}, we apply \cref{lem:2} to the case where $\rf$ is the block $\ell^1$-penalty. To that end we first compute its Fenchel conjugate. Recall that subsets $\La_\kk$ are assumed to be disjoint. 

\begin{lemma} The\label{lem:3} Fenchel conjugate of  the functional
$\rf\skl{\beta}
	=
	\sum_{\kk \in \KK}
	w_\kk \norm{\base_\kk\beta}_2 $
is given by
\begin{equation*}
\rf^*\kl{\mu} =
\begin{cases}
	0  & \text{ if }
	\displaystyle \max_{\kk\in \KK}
	\frac{\norm{\base_\kk\mu}_2}{w_\kk} \leq 1
	\\
	\infty & \text{ otherwise}\,.
\end{cases}
\end{equation*}
\end{lemma}

\begin{proof}
This can be verified by straightforward computation.
\end{proof}

We are now ready to prove  \cref{thm:group-lasso}.

\begin{proof}[Proof of \cref{thm:group-lasso}]
\cref{lem:2} implies that
$\beta_\star\in \R^p$ is a minimizer of the adaptive group lasso in
\eqref{eq:group-lasso} if and only if it is a minimizer of the functional
$ \frac{1}{2} \snorm{\X \beta}_2^2  +  \rf^*\skl{
\X^{\mathsf{T}}\skl{Y-\X  \beta }}$ with the particular  choice
$\rf \skl{\beta}
	=
	\gamma \sum_{\kk \in \KK}
	w_\kk \snorm{\base_\kk\beta}_2 $.
According to \cref{lem:3} this is equivalent to the fact that
$\beta_\star$ minimizes $ \frac{1}{2} \snorm{\X  \beta}_2^2$
under the constraint
\begin{equation*} \label{eq:block-dual}
	\max_{\kk\in\KK}
	\frac{\snorm{ \base_\kk \X^{\mathsf{T}} \kl{Y- \X  \beta}  }_2}{w_\kk}
	\leq   \gamma
	\;.
\end{equation*}
This means that $\beta_\star$ is the unique
minimizer of \eqref{eq:mre-gl} as we intended to show.
\end{proof}

\begin{proof}[Proof of \cref{cor:lasso-dual}]
By specializing $\kl{\phi_\la}_{\la \in \La}$ to the standard basis in $\R^n$, letting the groups consist of single elements, and taking all weights  equal to one, the constraint in \cref{thm:group-lasso} reduces to $\snorm{\X^{\mathsf{T}} \skl{Y-\X\beta}}_\infty \leq \gamma$.  Further, in such a situation the adaptive group lasso reduces to the standard lasso. Hence the claim follows from \cref{thm:group-lasso}.
\end{proof}

\begin{proof}[Proof of \cref{thm:bt-gl}]
According to \cref{thm:soft-block}, the estimators in~\eqref{eq:block-st2} and \eqref{eq:mre-gl} coincide. According to \cref{thm:group-lasso}, the estimators in \eqref{eq:mre-gl} and \eqref{eq:group-lasso} coincide.
\end{proof}

\begin{proof}[Proof of \cref{th:rvnem}]
It follows immediately from \cref{thm:lagrange,thm:lagrange-b}.
\end{proof}


\end{document}